\documentclass[a4paper,UKenglish,cleveref,autoref]{lipics-v2021}
\usepackage{mathtools,etoolbox,microtype}
\usepackage{tikz,subcaption,xspace,float}

\crefname{claim}{claim}{claims}
\Crefname{claim}{Claim}{Claims}

\def\squareforqed{$\Box$}
\def\QED{\ifmmode\squareforqed\else{\unskip\nobreak\hfil
\penalty50\hskip1em\null\nobreak\hfil\squareforqed
\parfillskip = 0pt\finalhyphendemerits = 0\endgraf}\fi}

\newcommand{\ignore}[1]{ }

\newcommand{\bst}[1]{\mbox{\boldmath \scriptsize$#1$}}

\newcommand{\cm}[1]{\mbox{\boldmath \scriptsize$#1$}}

\newcommand{\bA}{\mbox{\boldmath$A$}}
\newcommand{\bB}{\mbox{\boldmath$B$}}
\newcommand{\bC}{\mbox{\boldmath$C$}}

\newcommand{\bG}{\mbox{\boldmath$G$}}
\newcommand{\bX}{\mbox{\boldmath$X$}}

\newcommand{\N}{{\mathbb N}}

\newcommand{\T}{\mathcal{T}}

\newcommand{\br}{\mbox{\boldmath$r$}}

\newcommand{\sbr}{{\cm r}}

\newcommand{\bsigma}{\mbox{\boldmath$\sigma$}}
\newcommand{\bgamma}{\mbox{\boldmath$\gamma$}}
\newcommand{\btau}{\mbox{\boldmath$\tau$}}

\newcommand{\btheta}{\mbox{\boldmath$\theta$}}

\newcommand{\bomega}{\mbox{\boldmath$\omega$}}
\newcommand{\brho}{\mbox{\boldmath$\rho$}}

\newcommand{\sbsigma}{{\bst \bsigma}}

\renewcommand{\lambda}{\mu}

\newcommand{\FMA}{\mathsf{FMA}}
\newcommand{\BFMA}{\mathsf{BFMA}}
\newcommand{\HRA}{\mathsf{HRA}}
\newcommand{\CFG}{\mathsf{CFG}}
\newcommand{\BinaryCFG}{\mathsf{Binary}\text{-}\CFG}

\newcommand{\nullcontrol}{\mathsf{NULL}}

\newcommand{\fullcontrol}{\mathsf{FC}}
\newcommand{\unrestrainedcontrol}{\mathsf{URC}}
\newcommand{\rowcontrol}{\mathsf{RC}}
\newcommand{\columncontrol}{\mathsf{CC}}

\newcommand{\atoms}{\text{\sc Atoms}}
\newcommand{\Gr}{\text{\sc Perm}}

\newcommand{\setof}[2]{\set{#1 \,  : \, #2}}
\newcommand{\set}[1]{\{ #1 \}}
\newcommand{\supp}[1]{\text{\sc supp}(#1)}

\newcommand{\tuple}[1]{( #1 )}
\newcommand{\letter}[1]{\langle #1 \rangle}
\newcommand{\dom}[1]{\text{\sc dom}(#1)}

\newcommand{\size}[1]{|#1|}
\newcommand{\card}[1]{\lvert #1\rvert}
\newcommand{\Par}[1]{\text{\sc Par}\left(#1\right)}

\newcommand{\Ratof}[1]{\text{\sc Rat}\left(#1\right)}
\newcommand{\Rat}{\text{\sc Rat}}

\newcommand{\zerovector}{\mathbf{0}}

\newcommand{\conf}[2]{[#1,#2]            }
\newcommand{\trans}[1]{ \xrightarrow{#1}  }
\newcommand{\transword}[1]{ \xrightarrow{#1}\mathrel{\vphantom{\to}^*}}
\newcommand{\langof}[1]{L(#1)}
\newcommand{\langofconf}[3]{L_{#1,#2}(#3)}

\newcommand{\langh}[1]{H_{#1}}
\newcommand{\langpar}[2]{L_{#1}(#2)}

\newcommand{\Reg}{\mathcal{R}}

\hideLIPIcs
\nolinenumbers
\hypersetup{hidelinks}
\title{Star Complexity of Parikh Images of Languages over Infinite Alphabets}
\titlerunning{Star Height of Parikh Images}
\author{Yoav {Danieli}}{Technion -- Israel Institute of Technology, Faculty of Computer Science, Haifa, Israel.}{yoavd@campus.technion.ac.il}{https://orcid.org/0000-0001-8774-064X}{}
\authorrunning{Y. Danieli}
\Copyright{Yoav Danieli}
\makeatletter
\patchcmd{\@ccsdescString}{;}{.}{}{}
\makeatother
\DeclareCaptionLabelFormat{bwboxed}{%
 \kern0.05em{\color{black}\rule{0.73em}{0.73em}}%
 \hspace*{0.67em}\bothIfFirst{#1}{~}#2}

\begin{document}
\maketitle
\setlength{\emergencystretch}{2em}

\begin{abstract}
Parikh's theorem identifies the commutative images of regular and
context-free languages over finite alphabets. We study the commutative images
of quasi-regular and quasi-context-free languages
over infinite alphabets, these
are the languages that are accepted by
the corresponding
automata models augmented with registers.
Over finite alphabets commutative images
of regular languages
are known to be precisely the semilinear sets,
over infinite alphabets the corresponding algebraic
sets are called rational sets of data vectors.
In contrast to semilinear sets, nesting
stars in rational sets
add expressive power
and thus star-height is a measure of
complexity for this sets.

It was previously established that
languages accepted by
one-register automata
or
one-register binary branching context-free
grammars have
rational commutative images.
We refine this result by proving that the star-height
of the commutative image of any one-register
automaton is at most two.
Using this extension, we
prove rationality for all one-register 
context-free grammars, with
arbitrary branching degree.

Furthermore, we show that quasi-regular languages
can have irrational commutative images
already for a language accepted by a three-register
automaton.
Finally, we show that quasi-context-free languages
can have commutative images that are different
from the commutative image of any quasi-regular language, where the separation
is witnessed again with three registers.
Thus we conclude that Parikh's theorem fails over infinite alphabets for automata with registers.
\end{abstract}

\clearpage
\tableofcontents
\clearpage

\section{Introduction}
\label{s: introduction}

{\em Finite-memory automata}~\cite{KaminskiF94} generalize
classical Rabin-Scott finite-state automata~\cite{RabinS59} to
{\em infinite} alphabets.
By equipping the automata with a finite set of registers,
which store letters from the infinite alphabet
during computation and
restricting the power of the automaton to
comparing the input letters with register contents
and copying input letters to registers only,
the automaton retains a fixed finite memory of input letters.
Consequently, the languages accepted by
finite-memory automata possess properties
similar to regular languages and are,
thus, termed {\em quasi-regular} languages.

Quasi-regular languages
provide a finite description
of computations over unbounded data domains, where the relevant
information is which values agree rather than the identities of the
values themselves.
Such models arise in the study of data words,
databases, and verification~\cite{Segoufin06,Kara16,ChenSW16}.
Moreover, these
languages are precisely the
languages accepted by
orbit-finite automata over equality atoms
(also known as nominal
automata)~\cite{BojanczykKL14}.

Context-free grammars and pushdown automata also admit extensions with
registers.
These two models remain equivalent in expressive power and
define a class strictly larger
than the quasi-regular
languages~\cite{ChengK98}.

Over a finite alphabet, the difference between regular and context-free
languages disappears when the order of letters is forgotten.
Parikh's theorem~\cite{Parikh66} states that every context-free language
has the same commutative image as some regular language.
These images
are exactly the semi-linear sets~\cite{EilenbergS69}.
For example, $\{a^nb^n:n\ge0\}$ is not regular, but it has the same
Parikh image as $(ab)^\ast$.

Hofman, Juzepczuk, Lasota, and Pattathurajan~\cite{HofmanJLP21}
introduced rational expressions for sets of data vectors using addition,
the Kleene star, and orbit-finite unions. Orbit-finite unions replace the
finite unions of the classical setting while respecting the symmetry
of the alphabet.
Every rational set of data vectors is the Parikh image
of a quasi-regular language. This suggests a program for extending
Parikh's theorem: prove that the Parikh images of context-free languages
over infinite alphabets are rational. Rationality was established for
one-register finite-memory automata and binary one-register context-free
grammars~\cite{HofmanJLP21}.
Lasota and
Pattathurajan~\cite{LasotaP21} subsequently proved rationality for
hierarchical register automata.

Rationality alone does not describe the complexity of these images.
The \emph{star-height} of a rational set is the least nesting depth of
the Kleene star in a rational expression defining it.
Over finite
alphabets, Parikh's theorem gives the universal bound one. 

We establish the following results.
\begin{enumerate}
\item
The commutative image of every one-register finite-memory automaton has
star-height at most two.
This bound is sharp and refines
the rationality theorem in~\cite[Theorem 6]{HofmanJLP21}.
\item
Every one-register context-free grammar has a rational commutative image,
without a restriction on its branching degree.
Nevertheless, for every
$n\ge0$ there is a binary one-register context-free grammar whose
commutative image has star-height exactly $n$.
Thus the bound for automata
does not extend to one-register grammars.
\item
The commutative image of every hierarchical register automaton
with $r$ registers has star-height at most $r+1$.
This refine
the rationality theorem in~\cite[Theorem 14]{LasotaP21}.
\item
A three-register finite-memory automaton can recognize a language
with an irrational commutative image.
This refutes the general rationality
conjecture in~\cite{HofmanJLP21,LasotaP21}.
\item
A three-register context-free grammar can generate a language whose
commutative image is not the commutative
image of any quasi-regular language.
Hence context-free grammars have strictly greater commutative expressive
power than finite-memory automata.
In other words,
Parikh's theorem fails over infinite
alphabets.
\end{enumerate}

\section{Preliminaries}
\label{s: of sets}

We use the following conventions.
\begin{itemize}
\item
$\atoms$ is an infinite set of \emph{atoms}. We write $a,b,c$ for its elements, sometimes with indices or primes.
\item
Infinite alphabets are denoted by uppercase Greek letters,
$\Sigma,\Gamma,\Theta$ and
finite alphabets are denoted by uppercase Latin letters
$D,H,K$.
\item
Words over infinite alphabets 
are written in bold lowercase Greek letters $\bsigma$,\,$\bgamma$,\,$\btheta$, etc.,
also sometimes indexed or primed,
and range over $\Sigma^\ast,\Gamma^\ast,\Theta^\ast$, respectively.
\item
Symbols occurring in a word denoted by a boldface letter
are written using
the same \emph{non-boldface} letter with an appropriate subscript.
\item
For a word
$\bsigma = \sigma_1 \sigma_2 \cdots \sigma_n \in \Sigma^\ast$,
the \emph{set of contents} of $\bsigma$ is defined by
\[
[\bsigma] = \{ \sigma_1 , \sigma_2 , \ldots , \sigma_n \}\, .
\]
\item
For a subset $A\subseteq\atoms$ and a positive integer $r$,
we write $A^{r_{\neq}}$ for the set of all $r$-tuples
of pairwise distinct elements of $A$
and $A^{(r)}=\binom{A}{r}$
for the set of all $r$-element subsets of $A$. We also write $A^{(\leq r)}=\bigcup_{j=0}^r A^{(j)}$.
\item
Variables $x,y,z$, sometimes indexed or primed,
range over $\atoms$.
\item We write $\neq(x_1,x_2,\ldots,x_n)$ for pairwise distinctness, that is,
$\bigwedge_{i<j}x_i\neq x_j$.
Similarly, for a set of variables $Y$, the formula $x\notin Y$
abbreviates $\bigwedge_{y\in Y}x\neq y$.
For finite sets
of variables $X$ and $Y$, the expression $X \cap Y = \emptyset$
is interpreted as $\bigwedge_{x\in X} x\notin Y$.
\item The length of a word $\bsigma\in \Sigma^\ast$
is denoted by $|\bsigma|$ and the cardinality of a finite set $X$
is denoted by $\card{X}$.
\end{itemize}

\subsection{Orbit-finite sets}

We follow the definitions and notation of~\cite{HofmanJLP21}.
For an extensive overview of orbit-finite sets,
we refer the reader to
Boja\'nczyk's manuscript \emph{Slightly Infinite Sets}~\cite{Bojanczyk19}.

Let $\Gr$ be the group of permutations of $\atoms$. A permutation acts
coordinatewise on tuples and words, fixes control states and labels, and
acts on sets by $\pi X=\{\pi x:x\in X\}$. For finite $S\subseteq\atoms$,
let $\Gr_S$ consist of the permutations fixing $S$ pointwise, called
\emph{$S$-permutations}. The set $S$ \emph{supports} $x$ if $\pi x=x$
for every $\pi\in\Gr_S$. We work with finitely supported objects, including
the elements of our sets. Each has a least support $\supp{x}$.
An object with empty support is \emph{equivariant}.

The \emph{$S$-orbit} of $x$ is $\{\pi x:\pi\in\Gr_S\}$. A set is
\emph{orbit-finite} if it is a finite union of $S$-orbits for some finite
$S$.
Enlarging a finite support preserves orbit-finiteness~\cite[Lemma 3.16]{Bojanczyk19}.
The \emph{equality type} of an atom tuple over $S$
is specified by an equivalence relation on the atom
positions and atoms in $S$.
Notice that
two tuples of the same length lie
in the same $S$-orbit exactly when their
equality types over $S$ agree.

\begin{example}\label{e: atoms star}
For $S=\emptyset$,
$\atoms^2$ has two orbits the diagonal
$\{(a,a):a\in\atoms\}$
and the off-diagonal
$\{(a,b):a\neq b\in\atoms\}$.
In general, $\atoms^k$ is orbit-finite
and its number of orbits is the number of equality types, for $S=\emptyset$
it is the $k$-th Bell number.
However, $\atoms^\ast$ is not orbit-finite
because words of different lengths
are in different orbits.
\end{example}

\begin{proposition}[{\cite[Lemma 3.24]{Bojanczyk19}}]\label{p: of closure}
Orbit-finite sets are closed under binary union, intersection, Cartesian
product and projection.
\end{proposition}

An indexed family $(X_i)_{i\in I}$ is supported by $S$ if $I$ is supported
by $S$ and $X_{\pi i}=\pi X_i$ for every $i\in I$ and $\pi\in\Gr_S$.
For an orbit-finite $I$ and a finitely supported family,
$\bigcup_{i\in I}X_i$ is called an \emph{orbit-finite union}.

\subsection{Data words and vectors}

From now on, alphabets are orbit-finite.
A word in $\Sigma^\ast$ is a \emph{data word} and
a subset of $\Sigma^\ast$ is a \emph{data language}.
A \emph{data vector} over $\Sigma$ is a function 
$v:\Sigma\to\N$ that is zero except at finitely many letters.
For a data vector $v$ over $\Sigma$, its \emph{domain} and \emph{size} are
\[
\dom v=\{\sigma\in\Sigma:v(\sigma)>0\}\, , \qquad
\size v=\sum_{\sigma\in\dom v}v(\sigma)\, .
\]
We call $v(\sigma)$ the \emph{value} of $\sigma$ in $v$.
We represent a data vector $v$ as a formal sum
\[
v_1 \sigma_1+v_2 \sigma_2+\cdots
+v_n \sigma_n\, ,
\]
where $\dom{v}=\{\sigma_1,\sigma_2,\ldots,\sigma_n\}$,
and $v_i=v(\sigma_i)$ for $i=1,2, \ldots,n$.
For $v_i = 1$, we write $\sigma_i$ for $1 \sigma_i$.
Sometimes, we group atoms into disjoint subsets and write
\[
v_1 K_1 + v_2 K_2+\cdots + v_n K_n \, ,
\]
where $\dom{v}=\bigcup_{i=1}^n  K_i$ is a disjoint union
and $v$ is constant on each $K_i$,
i.e., $v(\sigma)=v_i$ for all $i=1,2,\ldots,n$ and $\sigma \in K_i$.

The zero data vector $\zerovector$ vanishes everywhere. The singleton data vector $\set{\sigma}$ takes value one at $\sigma$ and zero elsewhere. We often write it as $1\sigma$ or simply $\sigma$.
Addition of data vectors is point-wise:
$(v + v^\prime)(\sigma) =
v(\sigma) + v^\prime(\sigma)$
for every $\sigma\in\Sigma$. 
The order on data vectors is also point-wise,
i.e., $v_1\leq v_2$
if and only if for every $\sigma\in\Sigma$,
$v_1(\sigma)\leq v_2(\sigma)$.
Subtraction is defined whenever $v^\prime\leq v$ by
$(v-v^\prime)(\sigma)=v(\sigma)-v^\prime(\sigma)$.
If $v^\prime\leq v$ and $\sigma\in\Sigma$
is a letter such that $v^\prime(\sigma)=v(\sigma)$,
we say that $\sigma$ is \textit{saturated} in the pair $v^\prime\leq v$.

Bounded size characterizes orbit-finiteness for finitely supported sets of vectors.
\begin{lemma}[{\cite[Lemma 1]{HofmanJLP21}}]
\label{l: b-o-f}
A finitely supported set $X$ of data vectors is orbit-finite if
and only if
${\setof{\size v}{v\in X}}\subseteq \N$ is bounded.
\end{lemma}

The \emph{Parikh vector} (or \emph{commutative vector})
of a word $\bsigma\in \Sigma^\ast$
is the data vector $\Par{\bsigma} : \Sigma\to\N$, where
$\Par{\bsigma}(\sigma)$ is the number
of appearances of a letter $\sigma\in\Sigma$ in $\bsigma$.
Since we write $\size {\bsigma}$ for the \emph{length} of $\bsigma$,
by definition, $\size{\Par{\bsigma}} = \size {\bsigma}$.
Two words $\bsigma,\bsigma^\prime\in\Sigma^\ast$ are
called \emph{Parikh-equivalent} if they have
the same Parikh vector.
Equivalently, one word can be obtained by rearranging the letters of the other.

\begin{example}
Over the English alphabet, the following phrases are
Parikh-equivalent\footnote{Spaces are ignored.}:
\[
\begin{gathered}
\text{\itshape nondeterministic finite automaton}\\
\text{\itshape one timid cat turns one infinite atom}
\end{gathered}
\]
\end{example}

The \emph{Parikh image} (or \emph{commutative image})
of a language $L \subseteq \Sigma^\ast$
is
$\Par L = \setof{ \Par w }{w\in L}$.
Two languages $L, L^\prime\subseteq \Sigma^\ast$
are \emph{Parikh-equivalent} if they have the same Parikh images:
$\Par{L} = \Par{L^\prime}$.
For a class of languages $\mathscr{L}$,
let $\Par{\mathscr{L}}=
\setof{ \Par L }{L\in \mathscr{L}}$.

\begin{example}
Let $\Sigma$ be any alphabet and
let $D=\{\sigma\sigma:\sigma\in\Sigma\}$.
The language of palindromes over
$\Sigma$ is Parikh-equivalent to
$D^\ast\cup\Sigma D^\ast$.
Indeed, their common commutative image consists of all
data vectors with at most one odd entry.
\end{example}


\subsection{Rational sets of data vectors}
\label{ss: rat set dv}

We consider sets of data vectors over
a fixed orbit-finite alphabet $\Sigma$.
For two sets $X,Y$ of data vectors,
their \emph{Minkowski sum} (or addition) is
\[
X + Y = \setof{x+y}{x\in X, y\in Y}.
\]
For a set of data vectors $X$, its \emph{Kleene star} is
\[
X^\ast = \setof{x_1+x_2+\cdots+x_n}{n\geq 0, x_1,x_2,\ldots,x_n\in X}\, ,
\]
i.e., $X^\ast$ contains all finite sums of elements of $X$.

We define \emph{rational sets} of data vectors
as the smallest class of 
sets of data vectors that contains $\{\zerovector\}$,
all singletons $\{1\sigma\}$, and is closed under addition, the Kleene star,
and orbit-finite unions.
In particular, the empty set, all finite sets, and all
orbit-finite sets of data vectors are rational.
Furthermore, for an orbit-finite set of data vectors $P$,
the set $P^\ast$ is rational.
Let $\Ratof{\Sigma}$ be the
set of all rational
sets of data vectors over $\Sigma$,
we denote it by $\Rat$
when $\Sigma$ is clear from context.

The \emph{star-height} of a rational set of data vectors is the smallest
nesting depth of the Kleene star in any expression defining it using the
above operations.
A rational set of star-height at most $1$ is called
\emph{semi-linear}.

\begin{example}
\label{e: vector height two}
Fix $a\in\atoms$ and let
\[
 K_a=\bigcup_{b\in\atoms\setminus\{a\}}
       \bigl(\{1a\}+\{1b\}\bigr).
\]
The set $\atoms\setminus\{a\}$ has support $S=\{a\}$
and has one orbit under $\{a\}$-permutations
and therefore it is an orbit-finite union.
Now, let $B_a=\{2a\}+(K_a)^\ast$,
where $\{2a\}=\{1a\}+\{1a\}$.
For distinct $a,b,c\in\atoms$,
the data vector $5a+2b+c\in B_a$
because $(a+b),(a+c)\in K_a$
and therefore $(2a)+(a+b)+(a+b)+(a+c)\in B_a$.
Moreover, the empty sum gives $2a\in B_a$.

Permutations preserve addition and star, so the family $(B_a)_{a\in\atoms}$
is equivariant, with $\pi B_a=B_{\pi(a)}$ for every $\pi\in\Gr$.
Likewise, $\pi(B_a+B_b)=B_{\pi(a)}+B_{\pi(b)}$.
Since addition is commutative, $B_a+B_b$ depends only on $\{a,b\}$.
Thus both unions
\[
 B=\bigcup_{a\in\atoms}B_a,
 \qquad
 P=\bigcup_{\{a,b\}\in\atoms^{(2)}}(B_a+B_b)
\]
are orbit-finite because the atoms form one orbit, as do their
two-element subsets. Both expressions have star-height at most one.

Finally, put $X=P^\ast+(B\cup P)$.
For instance, $2a+2b\in P$ for
$a\ne b$, so $(2a+2b)+2a=4a+2b\in X$.
The star on $P$ adds at most
one level of nesting, hence $X$ is rational of star-height at most two.
\end{example}

\begin{remark}
\label{r: rat fin ab}
Over finite alphabets, every
rational set has star-height at most $1$,
and hence rationality
coincides with semi-linearity~\cite{EilenbergS69}.
This equivalence fails for
infinite alphabets,
the set $X$ from~\Cref{e: vector height two}
is not semi-linear
~\cite[Lemma 10]{HofmanJLP21}.
\end{remark}

\section{Finite-memory models and statements of our main theorems}
\label{s: fma cfg prelim}

Throughout this paper, input alphabets are of the form
$\Sigma=H\times\atoms$, where $H$ is a finite set,
called the set of \emph{labels}.
A letter $\sigma\in \Sigma$ is written
$\letter{h,a}$, where $h\in H$ is the \emph{label} of $\sigma$
and $a\in\atoms$ is the \emph{atom} component.

\subsection{Finite-memory automata}

We use the finite-memory model of~\cite{KaminskiF94}
augmented with the equality constraints of~\cite{HofmanJLP21}
and nondeterministic register updates of~\cite{KaminskiZ10}.

\begin{definition}
\label{d: fma}
An \emph{$r$-register finite-memory automaton}~($\FMA_r$)
is a tuple
\[
\bA=\langle S,s_0,F,H,D,\br_0,\mu\rangle
\]
with the following components.
\begin{itemize}
\item $S$ is a finite set of \emph{states}, $s_0\in S$ is the \emph{initial state}, and $F\subseteq S$ is the set of \emph{accepting states}.
\item $H$ is a finite set of \emph{labels}.
\item $D \subset \atoms$ is a finite set
of {\em distinguished} atoms (constants).
\item $\br_0=
\tuple{r_{0,1},r_{0,2},\ldots,r_{0,r}}\in \atoms^{r_{\neq}}$
specifies the initial register assignment,
where $r_{0,i}$ is the initial value of register $i$,
for $i=1,2,\ldots,r$.
\item $\mu$ is a finite set of transition rules of the form
\begin{equation}
\label{eq: trans}
s\trans{h,\varphi}s^\prime \, ,
\end{equation}
where $s,s^\prime\in S$, $h\in H$,
and the \emph{transition constraint}
$\varphi$
is a Boolean combination of equalities involving the variables
\[
x_1,x_2,\ldots,x_r,\quad y, \quad x_1^\prime,x_2^\prime,\ldots,x_r^\prime,
\]
and symbols from $D$.
\end{itemize}
\end{definition}

Throughout, a subscript $r$ on a model name denotes a bound of $r$ registers.
Omitting it permits any finite register bound. Thus $\FMA$ is the union of
the classes $\FMA_r$ over all finite register bounds $r$.

Write $\Reg=\atoms^{r_{\neq}}$ for the register valuations. A \emph{configuration} $\conf s{\br}$ consists of a state and a valuation. It is accepting when $s\in F$, and the initial configuration is $\conf{s_0}{\br_0}$. In a transition, $x_i$ and $x_i'$ denote the old and new values of register $i$, and $y$ denotes the input atom.
A transition rule~\eqref{eq: trans} and
atoms $a_1,a_2,\ldots,a_r$,
$b$,
and $a_1^\prime,a_2^\prime,\ldots,a_r^\prime$
that satisfy $\varphi$, with each register tuple in $\Reg$, induce a transition
\begin{equation*}
\conf{s}{\tuple{a_1,a_2,\ldots, a_r}}
\trans{\letter{h,b}}
\conf{s^\prime}{\tuple{a_1^\prime,a_2^\prime,\ldots, a_r^\prime}}
\, .
\end{equation*}

A \emph{run} of $\bA$ on a word
$\bsigma=\letter{h_1,b_1}\letter{h_2,b_2}\cdots \letter{h_n,b_n}
\in \Sigma^\ast$ is any sequence of configurations
\begin{align*}
{\conf {s_0} {\br_0}}
\trans{\letter{h_1,b_1}}
{\conf {s_1}{\br_1}}
\trans{\letter{h_2,b_2}}\conf{s_2}{\br_2}
\cdots\conf{s_{n-1}}{\br_{n-1}}
\trans{\letter{h_n,b_n}} 
\conf{s_n}{\br_n}.
\end{align*}
We abbreviate this by
$\conf {s_0} {\br_0}
\transword{\sbsigma} \conf{s_n}{\br_n}$.

A run is accepting if its last state is accepting. We write $\langof{\bA}$ for the language of words admitting an accepting run. Such languages are called \emph{quasi-regular}.
For configurations
$c,c^\prime$,
we write $\langofconf{c}{c^\prime}\bA$
for the set of all words admitting a run starting
in $c$ and ending in
$c^\prime$.
In particular,
\begin{equation}
\label{eq: l as lc}
\langof \bA =
\bigcup_{s \in F,\sbr\in\Reg}
\langofconf {\conf{s_0}{\sbr_0}} {\conf{s}{\sbr}} \bA
\, .
\end{equation}

\begin{proposition}[Equivariance of $\FMA$]
\label{p: permutation word fma}
Let $\alpha$ be a $D$-permutation of $\atoms$,
$\bsigma \in \Sigma^\ast$, and
let $c$ and $c^\prime$ be configurations of $\bA$ such that
$c \transword{\sbsigma} c^\prime$.
Then
$\alpha(c)
\transword{\alpha(\sbsigma)}
\alpha(c^\prime)$.
\end{proposition}

\begin{example}[{\cite[Example 1]{KaminskiF94}}]
\label{e: two}
Consider an automaton $\bA\in\FMA_1$,
over the alphabet $\Sigma=\atoms$
(the set of labels $H$ is treated as a singleton),
as shown in~\Cref{fig: aut l diff}.
The symbol $\top$ denotes
the formula that is always true.

The automaton stores a letter at its first chosen position and tests for it at a later position.
Hence it recognizes exactly the words with a repetition,
\[
\langof{\bA} = 
\{ a_1 a_2 \cdots a_n : \ \mbox{there exist} \
1 \leq i < j \leq n \ \mbox{such that} \ a_i = a_j \}\, .
\]
Its Parikh image, $\Par{\langof{\bA}}$, consists
of all data vectors which have a value of at least $2$
for some atom:
\[
\Par{\langof{\bA}}=
\set{v:
\mbox{there is an atom }a
\ \mbox{such that} \
v(a)\geq 2}\, ,
\]
or, equivalently,
\[
\Par{\langof{\bA}}= \bigcup_{a\in\atoms}2a+\atoms^\ast\, .
\]
\end{example}

\begin{figure}[h]
\begin{center}
\begin{tikzpicture}[scale=0.1]
\tikzstyle{every node}+=[inner sep=0pt]
\draw [black] (6.4,-28.21) -- (15.35,-28.29);
\fill [black] (15.35,-28.29) -- (14.55,-27.78) -- (14.55,-28.78);
\draw [black] (18.4,-28.2) circle (3);
\draw (18.4,-28.2) node {$s_0$};
\draw [black] (39.2,-28.3) circle (3);
\draw (39.2,-28.3) node {$s$};
\draw [black] (60.3,-28.2) circle (3);
\draw (60.3,-28.2) node {$f$};
\draw [black] (60.3,-28.2) circle (2.4);
\draw [black] (17.077,-25.52) arc (234:-54:2.25);
\draw (18.4,-20.95) node [above] {$\top$};
\fill [black] (19.72,-25.52) -- (20.6,-25.17) -- (19.79,-24.58);
\draw [black] (21.4,-28.21) -- (36.2,-28.29);
\fill [black] (36.2,-28.29) -- (35.4,-27.78) -- (35.4,-28.78);
\draw (28.8,-28.85) node [below] {$y=x_1^\prime$};
\draw [black] (42.2,-28.29) -- (57.3,-28.21);
\fill [black] (57.3,-28.21) -- (56.5,-27.72) -- (56.5,-28.72);
\draw (49.75,-29.55) node [below] {$x_1=y$};
\draw [black] (37.877,-25.62) arc (234:-54:2.25);
\draw (39.2,-21.05) node [above] {$x_1=x_1^\prime$};
\fill [black] (40.52,-25.62) -- (41.4,-25.27) -- (40.59,-24.68);
\draw [black] (58.977,-25.52) arc (234:-54:2.25);
\draw (60.3,-20.95) node [above] {$\top$};
\fill [black] (61.62,-25.52) -- (62.5,-25.17) -- (61.69,-24.58);
\end{tikzpicture}
\end{center}
\centering
\caption{A visuel representation of $\bA$.}
\label{fig: aut l diff}
\end{figure}

\begin{example}
\label{e: l first}
The language
\[
L_{\mathrm{first}}=
\{ \sigma_1 \sigma_2 \cdots \sigma_n : n\geq 2,\ \mbox{and for all} \
1 < i  \leq n, \ \sigma_1 \neq \sigma_i \}\, ,
\]
consists of all words whose first letter does not repeat.
Its reversal is the language that consists
of all words in which the last letter does not appear
earlier, namely,
\[
L_{\mathrm{last}}=
\{ \sigma_1 \sigma_2 \cdots \sigma_n : n\geq 2,\ \mbox{and for all} \
1 \leq i  < n, \ \sigma_n \neq \sigma_i \}\, .
\]
These languages
are accepted by automata $\bB$ and $\bC$
shown
below.
Furthermore, these languages are Parikh-equivalent.
Their Parikh image consists of all data vectors $v$
of size at least two such that $v(a)=1$ for some atom $a\in\atoms$,
\[
\Par{L_{\mathrm{first}}}=\Par{L_{\mathrm{last}}}=
\bigcup_{a\in\atoms}\bigl(a+(\atoms\setminus\set{a})+(\atoms\setminus\set{a})^\ast\bigr)\, .
\]

\begin{figure}[htbp]

\centering
\begin{subfigure}[t]{0.48\textwidth}
\centering
\begin{tikzpicture}[scale=0.1]
\tikzstyle{every node}+=[inner sep=0pt]
\draw [black] (6.4,-28.21) -- (15.35,-28.29);
\fill [black] (15.35,-28.29) -- (14.55,-27.78) -- (14.55,-28.78);
\draw [black] (18.4,-28.2) circle (3);
\draw (18.4,-28.2) node {$s_0$};
\draw [black] (39.2,-28.3) circle (3);
\draw (39.2,-28.3) node {$s$};
\draw [black] (60.3,-28.2) circle (3);
\draw (60.3,-28.2) node {$f$};
\draw [black] (60.3,-28.2) circle (2.4);
\draw [black] (21.4,-28.21) -- (36.2,-28.29);
\fill [black] (36.2,-28.29) -- (35.4,-27.78) -- (35.4,-28.78);
\draw (28.8,-30.75) node [below] {$y= x_1^\prime$};
\draw [black] (42.2,-28.29) -- (57.3,-28.21);
\fill [black] (57.3,-28.21) -- (56.5,-27.72) -- (56.5,-28.72);
\draw (49.75,-30.75) node [below] {$x_1=x_1^\prime\neq y$};
\draw [black] (58.977,-25.62) arc (234:-54:2.25);
\draw (60.3,-21.05) node [above] {$x_1=x_1^\prime\neq y$};
\fill [black] (61.62,-25.62) -- (62.5,-25.27) -- (61.69,-24.68);
\end{tikzpicture}
\caption{Automaton $\bB$.}
\label{fig: aut l first}
\end{subfigure}
\hfill
\begin{subfigure}[t]{0.48\textwidth}
\centering
\begin{tikzpicture}[scale=0.1]
\tikzstyle{every node}+=[inner sep=0pt]
\draw [black] (6.4,-28.21) -- (15.35,-28.29);
\fill [black] (15.35,-28.29) -- (14.55,-27.78) -- (14.55,-28.78);
\draw [black] (18.4,-28.2) circle (3);
\draw (18.4,-28.2) node {$s_0$};
\draw [black] (39.2,-28.3) circle (3);
\draw (39.2,-28.3) node {$s$};
\draw [black] (60.3,-28.2) circle (3);
\draw (60.3,-28.2) node {$f$};
\draw [black] (60.3,-28.2) circle (2.4);
\draw [black] (21.4,-28.21) -- (36.2,-28.29);
\fill [black] (36.2,-28.29) -- (35.4,-27.78) -- (35.4,-28.78);
\draw (28.8,-30.75) node [below] {$y\neq x_1^\prime$};
\draw [black] (42.2,-28.29) -- (57.3,-28.21);
\fill [black] (57.3,-28.21) -- (56.5,-27.72) -- (56.5,-28.72);
\draw (49.75,-31.75) node [below] {$x_1=y$};
\draw [black] (37.877,-25.62) arc (234:-54:2.25);
\draw (39.2,-21.05) node [above] {$x_1=x_1^\prime\neq y$};
\fill [black] (40.52,-25.62) -- (41.4,-25.27) -- (40.59,-24.68);
\end{tikzpicture}
\caption{Automaton $\bC$.}
\label{fig: aut l last}
\end{subfigure}

\caption{Automata for $L_{\mathrm{first}}$ and $L_{\mathrm{last}}$.}\label{fig: auto first last}\label{fig: aut l}
\end{figure}

\end{example}

\begin{remark}
\label{r: guessing}
Definition~\ref{d: fma} permits nondeterministic register updates
(also called \emph{guessing})
introduced in~\cite{ChengK98,KaminskiZ10}.
On a step reading $\letter{h,b}$, a machine without guessing can put
in a register only an old register value, $b$, or a distinguished atom.
A guess allows updating a register
with a value outside this set.
Permitting guesses
strictly increases expressive power.
For instance,
the language $L_{\mathrm{last}}$ from
\Cref{e: l first} cannot be recognized
without guessing~\cite[Example 4]{KaminskiF94}.
\end{remark}

\begin{remark}
\label{r: variants}
Definition~\ref{d: fma} additionally stipulates that registers
have predetermined initial values,
that there is a single initial state,
and that at every configuration all
registers store distinct atoms.

These design choices are not unique.
Several alternative formulations
appear in the literature.
For instance, register values
may be initialized nondeterministically,
or may start empty until assigned during a run,
or they may be erased and rewritten~\cite{MurawskiRT15}.
Some variants allow
the same atom to appear in multiple registers simultaneously~\cite{KaminskiF94,NevenSV04}.
These conventions do not change
the expressive power
of the model
\end{remark}

Labels and distinguished atoms can serve in place of one another under an
encoding. Choose distinct tag atoms $d_h$ for $h\in H$ and encode
$\letter{h,a}$ by the two successive atoms $d_h a$. Conversely, the atoms
of $D$ can be kept in dedicated registers initialized to their values, so
a test against a constant becomes a register test. We allow both labels
and distinguished atoms for generality and to retain the conventions
of~\cite{HofmanJLP21,KaminskiF94,NevenSV04}.

\begin{definition}
\label{d: hra}
An \emph{$r$-register hierarchical register automaton}~($\HRA_r$), $r\ge1$,
is an $\FMA_r$ as in \Cref{d: fma}
with $D=\emptyset$
whose transition
constraints have the form~\cite[Eq.~(5)]{LasotaP21}
\[
\varphi=\left(\bigwedge_{j<i}x'_j=x_j\right)\land\psi(x_1,x_2,\ldots,x_i,y,x'_i),
\]
where $i\in\{1,2,\ldots,r\}$
and $\psi$
is a Boolean combination of equalities
among the displayed variables.
We then say that $\varphi$ has \emph{level $i$}.
\end{definition}

\begin{example}\label{e:hra2-not-bfma1}
Consider the following language
\[
 L=\{ab^n a\in\atoms^\ast:
\ a\ne b,\ n\ge1\}.
\]
This language is in $\HRA_2$
by storing $a$ in register one
and the first
$b$ in register two.
However, it is not in $\FMA_1$
because after reading $ab$ the automaton must
forget either $a$ or $b$.
Therefore, $\HRA$ recognize strictly more
language than $\FMA_1$.
\end{example}

\begin{example}\label{e: hra separation}
Consider the following langueage
\[
 L=\{a_1a_2\cdots a_{3n}\in\atoms^\ast:
 n\ge1,\ 
        \neq(a_i,a_{i+1},a_{i+2})\text{ for }1\le i\le3n-2\}.
\]
It can be accepted by a two-register finite-memory
automaton that always store the last two atoms,
verify that each next atom is different from
them and update the registers.
However, no $\HRA$ recognizes $L$~\cite[Example~1 and Proposition~12]{LasotaP21}.
Thus hierarchical register automata recognize a proper subclass of the
quasi-regular languages.
\end{example}

\subsection{Block finite-memory automata}
\label{s: bfma}

In this section, we introduce a new model that generalizes $\FMA$ by allowing transitions to be labeled by finite strings
(blocks) rather than single letters.
Over finite alphabets,
this model is known as a \emph{generalized automaton}
\cite[Chapter 7, Section 10]{Eilenberg74},
\cite[Definition 2.2.1]{LewisP81}.

\begin{definition}
\label{d: bfma}
An \emph{$r$-register $k$-block finite-memory automaton}~($\BFMA_r$)
has the tuple of components in \Cref{d: fma}. Its transitions have the form
\begin{equation}
\label{eq: btrans}
s\trans{h_1h_2\cdots h_p,\varphi}s^\prime \, ,
\end{equation}
where $s,s^\prime\in S$,
$p$ is a nonnegative integer less than or equal to $k$,
$h_1,h_2,\ldots,h_p\in H$,
and the transition constraint
$\varphi$
is a Boolean combination of equalities involving the variables
$x_1,x_2,\ldots,x_r$ (current register values),
$y_1,y_2,\ldots,y_p$ (input atoms in the block),
$x_1^\prime,x_2^\prime,\ldots,x_r^\prime$ (next register values),
and symbols from $D$.
\end{definition}

In $\BFMA_r$, the block bound $k$ is fixed for each automaton and may be any positive integer, independently of $r$.
A configuration of $\bA$ is a pair $\conf{s}{\br}$
as in standard $\FMA$.
For atoms
$a_1,a_2,\ldots,a_r$,
$b_1,b_2,\ldots,b_p$,
$a_1^\prime,a_2^\prime,\ldots,a_r^\prime$
such that
\[
(a_1,a_2,\ldots,a_r,b_1,b_2,\ldots,b_p
,a_1^\prime,a_2^\prime,\ldots,a_r^\prime)
\models \varphi
\, ,
\]
where the old and new register tuples belong to $\Reg$, rule~\eqref{eq: btrans} gives the transition

\begin{equation}
\label{eq: bfma conf trans}
{\conf{s}{\tuple{a_1,a_2,\ldots,a_r}}}
\trans{\letter{h_1,b_1}\letter{h_2,b_2}\cdots \letter{h_p,b_p}}
{\conf{s^\prime}{\tuple{a_1^\prime,a_2^\prime,\ldots,a_r^\prime}}}\, .
\end{equation}
If $p=0$,~\eqref{eq: bfma conf trans} is an $\varepsilon$-transition.

A \emph{run} of $\bA$ over a word
$\bsigma \in \Sigma^\ast$ is
a sequence
of configurations $c_0,c_1,\ldots,c_m$
over a decomposition of $\bsigma$,
$\bsigma = \btau_1 \btau_2 \cdots \btau_m$,
such that
\begin{align*}
c_0 \trans{\btau_1} c_1 \trans{\btau_2} c_2\cdots c_{m-1}
\trans{\btau_m} c_m.
\end{align*}
Accepting runs, the language of the automaton,
and the language between configurations
are defined exactly as for standard $\FMA$.

\begin{example}
\label{e: l four}
Block transitions allow
simultaneous comparison
of several values.
For instance,
the language
\[
L_{\mathrm{four}}=\{abcd\in\atoms^4:\ \neq(a,b,c,d)\}
\]
is accepted by a $\BFMA_1$ with block bound two
depicted in~\Cref{fig: aut four diff}.

However, $L_{\mathrm{four}}$ cannot be recognized by $\FMA_1$.
Indeed, for distinct atoms $a,b,c,d$,
consider a one-register automaton run over the word $abcd$.
After reading $ab$ the
register stores at most one atom.
Hence some $x\in\{a,b\}$ and some
$y\in\{c,d\}$ are both absent from the register.
Transposing $x$ and
$y$ throughout the remaining run fixes
the configuration after $ab$
and preserves every transition constraint.
The resulting run therefore
accepts a word in which $x$ occurs in both halves, a contradiction.
\end{example}

\begin{figure}[htbp]
\centering
\begin{tikzpicture}[>=stealth,
  state/.style={circle,draw,minimum size=7mm,inner sep=1pt}]
\node[state] (s0) at (0,0) {$s_0$};
\node[state] (s1) at (3.2,0) {$s_1$};
\node[state] (s2) at (8.0,0) {$s_2$};
\node[state,double] (f) at (11.2,0) {$f$};
\draw[->] (-1,0) -- (s0);
\draw[->] (s0) -- node[below=3pt] {$y_1=x_1^\prime$} (s1);
\draw[->] (s1) -- node[below=3pt] {$\neq(x_1,y_1,y_2,x_1^\prime)$} (s2);
\draw[->] (s2) -- node[below=3pt] {$x_1=y_1$} (f);
\end{tikzpicture}
\caption{A $\BFMA_1$ with block bound two for $L_{\mathrm{four}}$.}
\label{fig: aut four diff}
\end{figure}

\begin{remark}
\label{r: bfma one geq fma one}
Every $\FMA_r$ is a $\BFMA_r$ with block bound one.
However, by Example~\ref{e: l four},
the class of languages recognized
by $\BFMA_1$
is strictly larger than the class recognized
by $\FMA_1$.
\end{remark}

\begin{lemma}\label{l: bfma simulation}
Every $\BFMA_r$ with block bound $k\ge1$ has an equivalent
$\FMA_{r+k-1}$ recognizing the same word language.
\end{lemma}
\begin{proof}
For a transition reading $p\ge1$ letters, retain the old register values
and read the first $p-1$ letters one at a time. Store each input atom not
already retained in an available register.
For a repeated atom, reuse its
register.
Thus at most $r+p-1\le r+k-1$ registers are needed.
On the last letter, the old values and earlier input atoms are available
in registers, the last input atom is the current letter, and the new
valuation is chosen by the transition. Substituting these values into
the block constraint gives a permitted transition constraint and simulates
exactly the original block transition.
Finally, eliminate
$\varepsilon$-transitions without adding registers by the construction
of~\cite[Lemma~5.3]{BojanczykKL14}, applied to the configuration space under
permutations fixing $D$.
\end{proof}

\subsection{Finite-memory context-free grammars}

We recall the definition of finite-memory
context-free grammars from~\cite{ChengK98}
with the equality
constraints of~\cite{BojanczykKL14} and constants.

\begin{definition}
\label{def: cfg gen}
An \emph{$r$-register finite-memory context-free grammar}~($\CFG_r$) is a system
\[
\bG=\langle V,H,K,P,S,\br_0\rangle,
\]
where
\begin{itemize}
\item $V$ is a finite set of \emph{variables} (nonterminals).
\item $H$ is a finite set of terminals, also
called \emph{labels}, and $V\cap H=\emptyset$.
\item $K\subset \Sigma$ is a finite set of distinguished symbols.
\item $P$ is a finite set of production rules of the form
\begin{equation}
\label{eq: cf prod}
A(X)\trans{\varphi} B_1(Y_1)B_2(Y_2)\cdots B_m(Y_m)\, ,
\end{equation}
where $A\in V$, $m$ is a nonnegative integer,
and 

$B_1,B_2,\ldots,B_m\in H\cup K\cup V$.
The transition constraint
$\varphi$
is a Boolean combination of equalities involving:
\begin{itemize}
\item the register variables
$X=\set{x_1,x_2,\ldots,x_r}$
(the register contents before the production),
\item for each $i$ such that $B_i\in H$, a single variable
$Y_i=\set{y_i}$
representing the atom produced,
\item for each $i$ such that $B_i\in V$, a set of $r$ variables
$Y_i={\set{y_{i,1},y_{i,2},\ldots,y_{i,r}}}$
representing the register contents associated
with the nonterminal $B_i$ after the application of the production,\footnote{
When $B_i \in K$, $Y_i$ is irrelevant.
We put $Y_i=\emptyset$ for uniformity.}
\end{itemize}
\item $\br_0\in\Reg$
is the initial register assignment.
\item $S\in V$ is the start symbol.
\end{itemize}
\end{definition}

The \emph{branching degree} (or \emph{arity}) of $\bG$ is the maximum length of a right-hand side. A grammar of branching degree two is called \emph{binary}. A \emph{binary one-register context-free grammar}~($\BinaryCFG_1$) is a binary grammar in $\CFG_1$.
A \emph{variable configuration} of $\bG$ is a pair
$\conf A {\br} \in V\times \Reg$,
consisting of a variable $A\in V$ and a register valuation
$\br \in \Reg$.
The variable configuration $\conf{S}{\br_0}$
is called the initial configuration.

Fix a production~\eqref{eq: cf prod} and an assignment satisfying its guard. Write $\br=(a_1,a_2,\ldots,a_r)\in\Reg$ for the parent valuation, $b_i$ for the atom at a label position, and $(c_{i,1},c_{i,2},\ldots,c_{i,r})\in\Reg$ for the tuple at a variable position. This instance gives the one-step derivation
\begin{equation*}
\conf{A}{\br}\Longrightarrow X_{1}X_{2}\cdots X_{m}
\,
,
\end{equation*}
Here $X_i=\letter{B_i,b_i}$ for a label $B_i\in H$ and $X_i=B_i$ for a literal $B_i\in K$. A variable $B_i\in V$ gives $X_i=\conf{B_i}{\tuple{c_{i,1},c_{i,2},\ldots,c_{i,r}}}$.

A derivation may replace one variable configuration anywhere in a word over $\Sigma\cup(V\times\Reg)$. Thus $\bX_1\conf A{\br}\bX_3\Longrightarrow\bX_1\bX_2\bX_3$ whenever $\conf A{\br}\Longrightarrow\bX_2$ is a rule instance.
As usual, the reflexive and transitive closure of
$\Longrightarrow$ is denoted by
$\Longrightarrow^\ast$.
The language $\langof{\bG}$ generated by
$\bG$ is defined by 
\[
\langof{\bG}=
\left\{
\bsigma\in\Sigma^{*}:
\conf{S}{\br_0}\Longrightarrow^\ast \bsigma\right\}
\,
.
\]
Any such language is called a \textit{quasi-context-free language}.

\begin{example}
There is a direct translation
from a $\BFMA$
to a \textit{right-regular}
finite-memory grammar
(each right-hand side contains terminals which are followed by at most one variable)
with the set of variables being the set of states of the automaton and
the following productions,
\[
s\trans{h_1h_2\cdots h_p,\varphi}s^\prime
\quad
\mapsto
\quad
s(X)\trans{\varphi}h_1(y_1)h_2(y_2)\cdots h_p(y_p) s^\prime(X^\prime)
\, .
\]
\end{example}

We identify variables when their equality is part of the constraint.
For instance, the production rule
$A(x)\rightarrow B(x)C(y)D(y)$
is an abbreviation for the production rule 
$A(x_1)\trans{x_1=y_1,y_2=y_3}B(y_1)C(y_2)D(y_3) $.

\begin{example}[{\cite[Example 3]{HofmanJLP21}}]
\label{ex: l mirror}
Consider the alphabet $\Sigma=\{l,r\}\times\atoms$.
Let $\bG_{\mathrm{mirror}}$ be
a $\CFG_1$ with $V=\{S\}$
and production rules
\[
P=\{S(x)\rightarrow l(x)S(y)r(x) ,
\quad S(x)\rightarrow l(x)r(x)\}\,.
\]
Then $\langof{\bG_{\mathrm{mirror}}}$
consists of all words of the following form,
\[
\letter{l,a_1}\letter{l,a_2}
\cdots \letter{l,a_n}
\letter{r,a_n}
\letter{r,a_{n-1}}
\cdots\letter{r,a_1}
\, .
\]

Observe that this language is Parikh-equivalent
to the language $L_{\mathrm{mirror.com}}$
which consists of all words of the form,
\[
\letter{l,a_1}\letter{r,a_1}
\letter{l,a_2}\letter{r,a_2}
\cdots 
\letter{l,a_n}\letter{r,a_n}
\,
.
\]
The latter language is recognizable by an $\FMA_1$.
Their common Parikh image is
\[
\letter{l,a}+\letter{r,a}+\left(\bigcup_{b\in\atoms} \letter{l,b}+\letter{r,b}\right)^\ast\, .
\]
\end{example}

\begin{example}
\label{ex: fixed atom pairs}
Over $\Sigma=\atoms$, let $\bG$ be the $\CFG_1$
with $V=\{S\}$ and rules
\[
 S(x)\trans{x\ne y}x\,S(x)\,y
 \qquad\text{and}\qquad
 S(x)\trans{x\ne y}x\,y.
\]
For $n\ge1$ and $b_1,b_2,\ldots,b_n\in\atoms\setminus\{a\}$,
$\langof{\bG}$ consists of the words $a^n b_1b_2\cdots b_n$.
Rearranging each word as $(ab_1)(ab_2)\cdots(ab_n)$ gives a
Parikh-equivalent language $L'$ recognized by an $\FMA_1$ that keeps
$a$ in its register.
With $K_a=\{a+b:b\in\atoms\setminus\{a\}\}$,
their common Parikh image is
\[
 \Par{\langof{\bG}}=\Par{L'}=K_a+K_a^\ast.
\]
\end{example}

\subsection{Parikh's theorem and main results}

Over a finite alphabet, Parikh's theorem
states that the commutative
images of regular and context-free
languages are exactly the semi-linear sets~\cite{Parikh66,EilenbergS69}.
We show that Parikh's theorem fails over infinite alphabets:
quasi-context-free languages 
are not
Parikh-equivalent to
quasi-regular languages.
A finite-memory context-free grammar with three registers suffices for this
separation.

\begin{theorem}
\label{t: par cf neq par fma}
Commutative images of quasi-context-free
languages form a strictly larger class than
the class of
commutative images of quasi-regular languages.
\end{theorem}

In fact, even the algebraic extension
of semilinear sets, the rational sets of data
vectors,
do not capture the commutative images
of automata over infinite alphabets,
already for finite-memory automata.
Again, three registers suffice for this separation.

\begin{theorem}
\label{t: par fma neq rat}
Commutative images of quasi-regular languages
are not always rational.
\end{theorem}

Despite these negative results for multi-register
automata we obtain positive
results in the restricted 
one-register variants.
We establish a sharp
universal upper bound of two on the star-height of commutative images of
one-register finite-memory automata.
This gives a precise
refinement of the
rationality result
in~\cite[Theorem~6 and Lemma~8]{HofmanJLP21}.

\begin{theorem}
\label{t: sh two}
Commutative images of one-register 
finite-memory automata
are of star-height at most two.
\end{theorem}

The bound is optimal, the example in~\cite[Lemma~10]{HofmanJLP21}
attains star-height two.
We extend the same universal bound to one-register block finite-memory
automata, which recognize a strictly larger class of languages
(cf.~\Cref{r: bfma one geq fma one}).

\begin{theorem}
\label{t: sh two block}
Commutative images of one-register block finite-memory automata
are of star-height at most two.
\end{theorem}

This extension is also the key to our grammar result, we use the block
automaton bound to prove rationality for all one-register finite-memory
context-free grammars, with arbitrary branching degree.
This extends the
binary-branching rationality theorem in~\cite[Theorem~7]{HofmanJLP21}.

\begin{theorem}
\label{t: rat one cfg}
Commutative images of one-register
finite-memory context-free grammars
are rational.
\end{theorem}

We further establish a full star-height hierarchy for one-register
context-free grammars: every nonnegative height is realized, already by
binary grammars.

\begin{theorem}
\label{t: high sh cfg}
For every $n\in\N$,
there is a one-register context-free grammar
generating a language whose commutative
image has star-height $n$.
\end{theorem}

Finally, for the model of hierarchical register
with $r$ registers we show an optimal bound
of star-height $r+1$.
This gives a precise
refinement of the
rationality result
in~\cite[Theorem~14]{LasotaP21}.

\begin{theorem}
\label{t: sh hra}
Commutative images of 
hierarchical register automata with $r$ registers
are of star-height at most $r+1$
and there are witnesses of such star-height for every $r\geq 1$.
\end{theorem}

Table~\ref{tab:parikh-spectrum} summarizes the exact star-heights attained
by each class of languages.

\begin{table}[htbp]
\centering\small
\setlength{\tabcolsep}{5pt}
\renewcommand{\arraystretch}{1.12}
\caption{Star-heights of Parikh images.}
\label{tab:parikh-spectrum}
\begin{tabular}{@{}p{.25\linewidth}p{.28\linewidth}p{.42\linewidth}@{}}
\hline
\raggedright Language class & \raggedright Heights or separation & \raggedright Source\tabularnewline
\hline
\raggedright Regular or context-free, finite alphabet & $0,1$ & \raggedright Parikh's theorem~\cite{Parikh66,EilenbergS69}\tabularnewline
\hline
\raggedright $\FMA_1$ & $0,1,2$ & \raggedright \Cref{t: sh two} 
\\
Rationality was shown in~{\cite[Thm.~6]{HofmanJLP21}}\tabularnewline
\raggedright $\BFMA_1$ & $0,1,2$ & \raggedright \Cref{t: sh two block}\tabularnewline
\raggedright $\HRA_r$ & $0,1,\ldots,r+1$ & \raggedright Theorem~\ref{t: sh hra}.
\\
Rationality was shown in~{\cite[Thm.~14]{LasotaP21}}\tabularnewline
\raggedright $\CFG_1$ & $\N$ & \raggedright \Cref{t: rat one cfg,t: high sh cfg}
\\
Rationality for $\BinaryCFG_1$
was shown in~\cite[Thm. 7]{HofmanJLP21} \tabularnewline
\hline
\raggedright $\FMA_3$ & \raggedright Some commutative images are irrational. & \raggedright \Cref{l: l fma3,t: irrational}\tabularnewline
\hline
\raggedright $\CFG_3$ & \raggedright Some commutative images are not commutative images of quasi-regular languages. & \raggedright \Cref{t: bg3 neq fma}\tabularnewline
\hline
\end{tabular}
\end{table}

\subsection{Organization of the paper}
Section~\ref{s: restricted models}
simplifies $\BFMA_1$ and $\CFG_1$ while
preserving their commutative images.
In Sections~\ref{a: restricted}
and~\ref{a: altering semilinearity}, we introduce the languages of
\emph{altering sets} as a new tool for analyzing commutative images of
one-register automata.
We establish semi-linearity of their Parikh
images and use it to prove
\Cref{t: sh two,t: sh two block}.

Section~\ref{a: rest cfg} then uses $\BFMA_1$ to extend the rationality result from $\BinaryCFG_1$ in~\cite{HofmanJLP21} to arbitrary branching
degrees, proving \Cref{t: rat one cfg}.

Section~\ref{s: lin form} develops linear forms as a method for proving
lower bounds.
In Section~\ref{s: cfg high sh}, we construct a family
of 
$\BinaryCFG_1$
and use the linear forms
to show that this family
realize any nonnegative star-height,
thus proving \Cref{t: high sh cfg}.
Similar constructions are presented
for $\HRA$ in Section~\ref{s: hra}
where we also establish the
optimal star-height bound $r+1$ for $\HRA_r$ and prove~\Cref{t: sh hra}.

Sections~\ref{s: par fma neq rat} and~\ref{a: par cfg neq par fma}
establish the main separations between
$\Rat$, $\Par{\FMA}$,
and $\Par{\CFG}$.
The first constructs
a quasi-regular language with an irrational commutative image,
proving \Cref{t: par fma neq rat}.
The second constructs a
quasi-context-free language whose commutative image is not that of
any quasi-regular language, proving \Cref{t: par cf neq par fma}.
Both proofs rely on \emph{commutative stability}, a new concept that
recovers structural information from commutative images.
We conclude in Section~\ref{s: conc}
with some remarks.

\section{Parikh-equivalent restricted models}\label{s: restricted models}

\subsection{Restricted block automata}\label{a: restricted rules}

\begin{definition}\label{d: restricted bfma}
Write $B(\btau)$ for the set of atoms occurring in a block $\btau$.
A $\BFMA_1$ is \emph{restricted} if every transition constraint has
one of the following forms:
\[
\begin{array}{ll}
\mathsf{PresEq} &
  \theta \land \gamma(x) \land x=x'
  \land B(\btau)\subseteq\{x\},\\
\mathsf{PresDiff} &
  \theta \land \gamma(x) \land x=x'
  \land B(\btau)\cap\{x\}=\emptyset,\\
\mathsf{UpDiff} &
  \theta \land x\ne x'
  \land B(\btau)\cap\{x,x'\}=\emptyset.
\end{array}
\]
Here $\theta$ is a equality type of the input variables over $D$:
it specifies which variables and constants are equal, with distinct constants
in distinct classes. The formula $\gamma(x)$ is a possibly empty conjunction
of tests $x=d$ or $x\ne d$, with $d\in D$.
The first two forms are called \emph{preserving}, and the third is called
\emph{updating}.
\end{definition}

\begin{lemma}\label{l: par res}
For every $\bA\in\BFMA_1$ and configurations $c,c'$ of $\bA$,
there are a restricted $\bB\in\BFMA_1$ and configurations $d,d'$ of
$\bB$ such that
\[
 \Par{\langofconf c {c'} {\bA}}
 =\Par{\langofconf d {d'} {\bB}}.
\]
\end{lemma}

\begin{example}
The rule $s\trans{h_1h_2h_3,\varphi}t$, where
\[
 \varphi=(x=y_1)\land(x'=y_2)\land\neq(y_1,y_2,y_3),
\]
is replaced by
\[
 s\trans{h_1,x=y_1=x'}u
  \trans{h_3,\neq(x,y_1,x')}v
  \trans{h_2,x=y_1=x'}t.
\]
The displayed rules abbreviate their refinements by equality
types on the input variables and $D$. For distinct atoms $a,b,c$,
this replaces the block
$\letter{h_1,a}\letter{h_2,b}\letter{h_3,c}$ by
$\letter{h_1,a}\letter{h_3,c}\letter{h_2,b}$, preserving its Parikh vector
and changing the register from $a$ to $b$ only in the middle step.
\end{example}

\begin{proof}[Proof of~\Cref{l: par res}]
Split each rule with constraint $\varphi$ into one rule for each
equality type $\theta$ of $x,y_1,\ldots,y_m,x'$ over $D$ that implies
$\varphi$. This is a finite split, and every satisfying assignment
realizes exactly one of these types. For a resulting rule
$s\trans{h_1h_2\cdots h_m,\theta}t$, partition its positions $M=\{1,2,\ldots,m\}$ into
\[
 I=\{i\in M:\theta\models y_i=x\},\quad
 J=\{i\in M\setminus I:\theta\models y_i=x'\},\quad
 K=M\setminus(I\cup J).
\]
For $U\subseteq M$, let $h_U$ list its labels in increasing position
order, and let $\theta|_U$ be the induced input type over $D$.
Using states $u,v$ private to this rule, replace it by
\[
 s\trans{h_I,\theta_I}u
  \trans{h_K,\theta_K}v
  \trans{h_J,\theta_J}t.
\]
Each new constraint includes its input type $\theta|_U$, with positions
renumbered in order. The outer rules have form $\mathsf{PresEq}$.
Their register tests against $D$ are exactly those prescribed by $\theta$
for the old and new register values, respectively, even when the
corresponding input block is empty. The middle rule has form
$\mathsf{PresDiff}$ if $\theta\models x=x'$ and $\mathsf{UpDiff}$ otherwise,
with no additional register tests. Every new block has length at most $m$.

Every original transition yields a run of the three rules by reading
the positions in the order $I,K,J$. Conversely, restore the input
letters of such a run to their original positions. Within each group
the input type is enforced. Across groups, the $I$ and $J$ atoms equal
their respective register values, whose mutual relation is tested by
the middle rule, and every $K$ atom differs from both. The input types
and the outer register tests enforce all comparisons with $D$. Thus the
restored assignment satisfies $\theta$ and has the same Parikh vector.
Let $\bB$ be the resulting automaton and $d,d'$ the copies of $c,c'$.
Since the intermediate states are private, every run between these
configurations decomposes into complete replacement paths. Expanding
or contracting them proves the two inclusions.
\end{proof}

\subsection{Restricted context-free grammars}\label{a: restricted grammars}

\begin{definition}\label{d: restricted cfg}
A $\CFG_1$ is \emph{restricted} if every production has a single terminal
or $\varepsilon$ on the right, or has the form
$A(x)\to B_1(y_1)B_2(y_2)\cdots B_m(y_m)$, $m\ge1$, with a constraint implying one of
\[
\begin{array}{ll}
 \mathsf{AllEq}
   &\displaystyle\bigwedge_{i=1}^{m}x=y_i,\\
 \mathsf{AllDiff}
   &\neq(x,y_1,y_2,\ldots,y_m).
\end{array}
\]
\end{definition}

\begin{lemma}\label{l: par res cfg}
Every $\CFG_1$ has a Parikh-equivalent restricted $\CFG_1$.
\end{lemma}

\begin{example}
Consider $A(x)\to B(y)C(x)E(y)F(z)$ with constraint $\neq(x,y,z)$.
Using fresh variables $U,V,W$, replace it by
\[
\begin{aligned}
 A(x)&\to C(x)U(x), & U(x)&\trans{\neq(x,y,z)}V(y)W(z),\\
 V(y)&\to B(y)E(y), & W(z)&\to F(z).
\end{aligned}
\]
Only the production for $U$ updates the register. The expansion groups
$B(y)$ and $E(y)$ together and preserves the multiset of all four
child configurations.
\end{example}

\begin{proof}[Proof of~\Cref{l: par res cfg}]
Retain productions with a single terminal or $\varepsilon$ on the right.
In every other production, replace each terminal occurrence $h(y)$ by a
fresh variable $T(y)$ with sole rule $T(y)\to h(y)$, and each literal
$k\in K$ by a fresh $T(x)$ with sole rule $T(x)\to k$, where $x$ is the
parent register value. This preserves the yield and the original constraint.
Split each resulting nonterminal production into its satisfying
equality types on the parent and child register variables.

Fix such a type $\theta$ for $A(x)\to B_1(y_1)B_2(y_2)\cdots B_m(y_m)$.
Let $I_0$ contain the positions equal to $x$ under $\theta$, and let
$I_1,I_2,\ldots,I_q$ be the remaining equality classes of positions,
with representatives $z_1,z_2,\ldots,z_q$. If $q=0$, retain the
production with its constraint. Otherwise introduce variables
$U,V_1,V_2,\ldots,V_q$ private to this production and type, and use
\[
\begin{aligned}
 A(x)&\to \Bigl(\prod_{i\in I_0}B_i(x)\Bigr)U(x),\\
 U(x)&\trans{\theta'}V_1(z_1)V_2(z_2)\cdots V_q(z_q),\\
 V_j(z_j)&\to \prod_{i\in I_j}B_i(z_j) \qquad (1\le j\le q).
\end{aligned}
\]
Products list positions in increasing order. Obtain $\theta'$ from
$\theta$ by substituting $x$ for $y_i$ on $I_0$ and $z_j$ for $y_i$
on $I_j$. Then $\theta'$ implies $\neq(x,z_1,z_2,\ldots,z_q)$.
Thus the middle rule has form $\mathsf{AllDiff}$ and the outer rules
have form $\mathsf{AllEq}$.

Each production instance expands to the same child configurations and
terminals up to rearrangement. Conversely, the private variables force
the displayed expansion, and $\theta'$ ensures that restoring the children
to their original positions gives an instance of the original production.
Expanding or contracting these occurrences in a complete derivation
therefore preserves the Parikh image from every original variable configuration.
\end{proof}

\section{Upper bounds for one-register automata}
\label{a: restricted}

By \Cref{l: par res}, it suffices to prove the upper bounds for restricted
automata. Throughout this section, fix a restricted $\bA\in\BFMA_1$
with state set $S$, distinguished atoms $D$, and block bound $k$.

For technical reasons,
we prove a refined version of \Cref{t: sh two block} stated
in \Cref{l: two_sh} below, which,
due to \eqref{eq: l as lc},
implies \Cref{t: sh two block}.

\begin{lemma}
\label{l: two_sh}
Let $\bA\in\BFMA_1$.
For every pair of configurations $c,c^\prime$,
the language $\langofconf c {c^\prime} \bA$
has a rational Parikh image
of star-height at most two.
\end{lemma}

\subsection{Substitution and star-height}\label{a: sub}

We recall a useful technique called
\emph{substitution} from~\cite{HofmanJLP21}.

Let $L$ be a language over an orbit-finite alphabet $\Sigma$,
and let $K = (K_\sigma)_{\sigma\in\Sigma}$ be a family
of languages over an alphabet $\Gamma$,
indexed by $\Sigma$,
such that the mapping $\sigma\mapsto K_\sigma$
is finitely supported.
The \emph{substitution} $L(K)$
is the language over $\Gamma$ obtained
by replacing each letter $\sigma_i$ of a word
$\sigma_1\sigma_2\cdots\sigma_n\in L$
with a word from $K_{\sigma_i}$:
\[
L(K)=\bigcup_{\substack{n\ge0\\\sigma_1\sigma_2\cdots\sigma_n\in L}}
K_{\sigma_1}K_{\sigma_2}\cdots K_{\sigma_n}.
\]

\begin{lemma}[{cf.~\cite[Lemma 5]{HofmanJLP21}}]\label{l:subst}
If $\Par L$ has star-height at most $h$ and every
$\Par{K_\sigma}$ has star-height at most $j$, then
$\Par{L(K)}$ has star-height at most $h+j$.
\end{lemma}
\begin{proof}
Choose an expression $E$ for $\Par L$ of star-height at most $h$,
and a finite support $S$ of $E$, $(K_\sigma)$, and the alphabets.
For each $S$-orbit representative $a$, choose an expression $E_a$
for $\Par{K_a}$ of star-height at most $j$. For $\sigma$ in the orbit of $a$, put
\[
 F_\sigma=\bigcup\{\pi E_a:\pi\in\Gr_S,\ \pi a=\sigma\}.
\]
For each $\sigma$, the displayed family is supported by
$S\cup\supp\sigma$ and orbit-finite by support enlargement. Thus
$F_\sigma$ defines $\Par{K_\sigma}$ with star-height at most $j$,
and $F_{\pi\sigma}=\pi F_\sigma$ for $\pi\in\Gr_S$.

Replace each letter $\sigma$ in $E$ by $F_\sigma$.
Structural induction verifies the substitution and its bound.
\emph{Base.} Letters have replacement star-height at most $j$,
while zero and the empty set are unchanged.
\emph{Step.} Independent replacement commutes with sums, unions,
and star. Sums and unions take the maximum operand bound, while
star adds one. Equivariance preserves orbit-finite unions.
The resulting expression therefore defines $\Par{L(K)}$
and has star-height at most $h+j$.
\end{proof}

\subsection{Runs with a fixed register value}\label{a: preserving runs}

For states $s_0, s_1\in S$ of $\bA$ and an atom $a\in\atoms$,
let $L_{s_0,a, s_1}$ be the language of all words
read by a run from configuration
$\conf {s_0} a$ to 
$\conf {s_1} a$
that use register-preserving transitions only
(thus, the register stores $a$ along the entire run).

\begin{lemma}\label{l:preserving-image}
The star-height of $\Par{L_{s,a,t}}$ is at most one for all $s,a,t$.
\end{lemma}

\begin{proof}
Regard the finite set $\mu$ of transition rules as an alphabet.
Define a finite automaton with states $S$ and an edge
$u\trans{e}v$ for each rule $e=(u,h_1h_2\cdots h_m,\varphi,v)\in\mu$.
Let $U_{s,t}$ be its language with initial state $s$ and accepting state $t$.
For each such rule, define
\[
 K_e(a)=\{\letter{h_1,b_1}\letter{h_2,b_2}\cdots\letter{h_m,b_m}:
 b_i\in\atoms,\ (a,b_1,b_2,\ldots,b_m,a)\models\varphi\}.
\]
Thus $K_e(a)$ contains precisely the words read by transitions induced
by $e$ with old and new register value $a$.

Concatenating these transitions, or decomposing a run into its transitions,
gives both inclusions in
\begin{equation}\label{eq:preserving-substitution}
 L_{s,a,t}=U_{s,t}(K(a)).
\end{equation}
Parikh's theorem gives star-height at most one for $\Par{U_{s,t}}$.
Each $K_e(a)$ is supported by $D\cup\{a\}$ and has bounded length,
so its Parikh image has star-height zero by \Cref{l: b-o-f}.
Apply \Cref{l:subst} to~\eqref{eq:preserving-substitution} to obtain
the bound $1+0=1$.
\end{proof}

\subsection{Altering paths}\label{a: ap as}\label{a: altering paths}

\begin{definition}\label{d: altering paths}
An \emph{altering path} is a word over the disjoint alphabets of markers
$\tuple{s,a,t}$, where $s,t\in S$ and $a\in\atoms$, and block letters
$\tuple{\btau}$, where $\btau\in\Sigma^\ast$ and $|\btau|\le k$.
It has the form
\[
 z_1\,\tuple{\btau_1}\,z_2\,\tuple{\btau_2}\cdots
 z_{n-1}\,\tuple{\btau_{n-1}}\,z_n,
 \qquad z_i=\tuple{s_i,a_i,t_i},\quad n\ge1,
\]
where, for every $1\le i<n$, an updating rule of $\bA$ induces
$\conf{t_i}{a_i}\trans{\btau_i}\conf{s_{i+1}}{a_{i+1}}$.
For configurations $c,c'$, let $P_{c,c'}$ be the language of these words
with $\conf{s_1}{a_1}=c$ and $\conf{t_n}{a_n}=c'$.
\end{definition}

\begin{lemma}\label{l:LF}
$\Par{P_{c,c'}}$ is semi-linear.
\end{lemma}

We precede the proof of \Cref{l:LF} by showing that it implies \Cref{l: two_sh}.

\begin{proof}[Proof of \Cref{l: two_sh} from \Cref{l:LF}]
Assume \Cref{l:LF}. We have
\begin{equation}\label{eq:run-substitution}
 L_{c,c'}=P_{c,c'}(K),\qquad
 K_{\tuple{s,a,t}}=L_{s,a,t},\quad K_{\tuple{\btau}}=\{\btau\}.
\end{equation}
Indeed, cutting a run around its updating transitions gives the marker
replacements in~\eqref{eq:run-substitution}. Conversely, runs chosen for
these markers concatenate with the updates at their recorded endpoints.
With no update, one marker represents the whole run.

The marker replacements have Parikh images of star-height at most one
by \Cref{l:preserving-image}, and the block replacements are singletons.
The family is supported by $D$. Applying \Cref{l:subst} to
\eqref{eq:run-substitution} bounds the star-height of the resulting image
by $1+\max\{1,0\}=2$.
\end{proof}

\subsection{Altering sets}\label{a: altering sets}

\begin{definition}\label{d: altering sets}
A \emph{connector} is a letter $\tuple{t,B,s}$ with $t,s\in S$ and
$B\in\atoms^{(k)}$. Connectors form an alphabet disjoint from the
markers. An \emph{altering set} is a word that alternates between
markers and connectors, beginning and ending with a marker:
\begin{equation}\label{eq:as}
 z_1g_1z_2g_2\cdots z_{n-1}g_{n-1}z_n,
 \qquad n\ge1.
\end{equation}
Write $z_i=\tuple{s_i,a_i,t_i}$ for $1\le i\le n$. For each $i<n$,
the connector between $z_i$ and $z_{i+1}$ must have the form
$g_i=\tuple{t_i,B_i,s_{i+1}}$, where
\[
 a_i\ne a_{i+1}
 \qquad\text{and}\qquad
 B_i\cap\{a_i,a_{i+1}\}=\emptyset.
\]
Let $Q_{c,c'}$ be the language of all such words with boundary
configurations $c=\conf{s_1}{a_1}$ and $c'=\conf{t_n}{a_n}$.
\end{definition}

The connector $g_i$ records the states before and after an update
and a $k$-element set $B_i$ from which the atoms of the updating block
may be drawn. Thus, once $S$, $k$, and the boundary configurations
are fixed, the language $Q_{c,c'}$ is independent of the transition
rules. Up to renaming states, only $|S|$ matters in the choice of $S$.
Therefore the family of altering sets
languages is automata-independent.

\begin{lemma}\label{l:LF2}
$\Par{Q_{c,c'}}$ is semi-linear.
\end{lemma}

Before proving \Cref{l:LF2}, we show how it implies \Cref{l:LF}.

\begin{proof}[Proof of \Cref{l:LF} from \Cref{l:LF2}]
A block \emph{matches} an updating rule if it has the labels prescribed
by the rule and satisfies the rule's equality type on the input
variables and $D$. For each marker $z$, put $K_z=\{z\}$. For each
connector, put
\begin{equation}\label{eq:connector-substitution}
 K_{\tuple{t,B,s}}
 =\bigl\{\tuple{\btau}:B(\btau)\subseteq B,\ 
   \btau\text{ matches an updating rule from }t\text{ to }s\bigr\}.
\end{equation}

We claim that $P_{c,c'}=Q_{c,c'}(K)$. For the inclusion from right to
left, replace each connector $g_i$ by any block letter in $K_{g_i}$.
Matching enforces the conditions on the input atoms. The remaining
conditions hold because $a_i\ne a_{i+1}$ and every input atom belongs
to $B_i$, which avoids both register values. Each replacement therefore
gives an updating transition between the two recorded configurations.
Conversely, for each updating block of an altering path, its atom set
has size at most $k$ and avoids both register values. Since $\atoms$ is
infinite, extend this set to a $k$-element set still avoiding those
values. The resulting connectors give a word of $Q_{c,c'}$ whose
substitution contains the original path.

For fixed $B$, the set in~\eqref{eq:connector-substitution} is finite:
the label set $H$ and the atom set $B$ are finite, and the block length
is at most $k$. Hence every replacement has a Parikh image of
star-height zero. The substitution family is supported by $D$, so
\Cref{l:subst,l:LF2} give a star-height bound of $1+0=1$ for
$\Par{P_{c,c'}}$.
\end{proof}

Together with the implication from \Cref{l:LF} to \Cref{l: two_sh},
this reduces the automaton bound to the semi-linearity of
$\Par{Q_{c,c'}}$, which we prove in the next section.

\section{Semi-linearity of altering sets}\label{a: altering semilinearity}

Our goal is \Cref{l:LF2}.
We show that every altering-set word is
Parikh-equivalent to a bounded word plus a sum of bounded remainders.
The main issue is to make deletion reversible: after shortening a word,
the same remainders must still be insertable.
A good control retains
finitely many witness joins that guarantee this property.

\subsection{Restrained sets}
\label{a: crossings}

Put $p=k+1$ and
$\Theta_p=\{(a,B):B\in\atoms^{(p)},\ a\notin B\}$.
For an altering-set word as in~\eqref{eq:as},
\[
 \bsigma=z_1g_1z_2g_2\cdots z_{n-1}g_{n-1}z_n\in Q_{c,c'},
 \qquad z_i=\tuple{s_i,a_i,t_i},\quad
 g_i=\tuple{t_i,B_i,s_{i+1}},
\]
the \emph{join at position $i$}, $1\le i<n$, is specified by the state
pair $(t_i,s_{i+1})$ and
\[
 (a_i,D_i)=(a_i,B_i\cup\{a_{i+1}\})\in\Theta_p.
\]
Call $a$ the \emph{row} and $B$ the \emph{column} of a pair $(a,B)$.
To \emph{cross}
$(a,B)$ and $(c,D)$ is to replace them by $(a,D)$ and $(c,B)$.
The crosses are legal exactly when $a\notin D$ and $c\notin B$.
We apply crossings to pairs at positions with the same state pair.

\begin{definition}
A set $A\subseteq\Theta_p$ is \emph{restrained} by $(z,W)\in\Theta_p$ if every $(a,B)\in A$ satisfies $a\in W$ or $z\in B$. It is \emph{unrestrained} if no such pair exists.
\end{definition}

\begin{remark}
\label{r: rest}
Every subset of a restrained set is also restrained
by the same pair, i.e., this property is downward-closed.
In particular, being unrestrained is upward-closed.
\end{remark}

\begin{example}
\label{e: ur set ex}
Let $p=1$ and $a,b$ and $c$ be distinct atoms. Consider
\[
 A=\{(a,\{b\}),(b,\{c\}),(c,\{a\})\}.
\]
For any proposed restraint
$(z,\{w\})$,
at least two rows differ from $w$.
Their columns are distinct singletons,
so they cannot both contain $z$.
Thus $A$ is unrestrained.
\end{example}

\begin{lemma}\label{l: unrest simple}
For $p\ge1$, every unrestrained $A\subseteq\Theta_p$ has an unrestrained
subset of size at most $b_p=(p+1)+p(p+1)^2$.
\end{lemma}
\begin{proof}
Choose pairs in $p+1$ distinct rows and let $U$ be their column union.
Otherwise at most $p$ rows would restrain $A$. For each $z\in U$,
consider the pairs whose columns avoid $z$. Keep one in row $z$ if
possible, and otherwise keep pairs in $p+1$ distinct rows. Such rows
must exist, since a $p$-set covering them and avoiding $z$ would restrain
$A$. The initial pairs defeat every restraint $(z,W)$ with $z\notin U$,
and the added pairs defeat those with $z\in U$. Their total number is
\[
 (p+1)+|U|(p+1)\le(p+1)+p(p+1)^2=b_p.
\]
\end{proof}

\begin{lemma}\label{l: unrest crossing}
Let $p\ge1$, $A\subseteq\Theta_p$ be unrestrained, and $(c,D)\in\Theta_p$.
There is a pair $(a,B)\in A$ 
that satisfies:
\begin{enumerate}
\item $(a,D),(c,B)\in\Theta_p$.
\item $A^\prime$ is unrestrained,
where
\[
A^\prime=(A\setminus\{(a,B)\})\cup\{(a,D),(c,B)\}\, .
\]
\item $A^\prime\cup \{(a,B)\}$ is also unrestrained.
\end{enumerate}
\end{lemma}

\begin{proof}
If row $c$ occurs, choose $a=c$. Crossing then gives $A\cup\{(c,D)\}$,
which is unrestrained by \Cref{r: rest}. Otherwise, since $(c,D)$ does not restrain $A$,
choose $(a,B)$ with $a\notin D$ and $c\notin B$.
Suppose $(z,W)$ restrains the set after crossing. It cannot restrain
the removed pair, so $a\notin W$ and $z\notin B$. The new pairs force
$z\in D$ and $c\in W$. Hence
\[
 W'=(W\setminus\{c\})\cup\{a\}
\]
is a $p$-set avoiding $z$. No old row equals $c$, so $(z,W')$ restrains
every old pair, including $(a,B)$, a contradiction.
Hence $A^\prime$ is unrestrained,
by~\Cref{r: rest} any set that
contains $A^\prime$ is also unrestrained.
\end{proof}

\subsection{Row and column controls}\label{a: con as}

\begin{definition}
\label{d: control witnesses}
\label{d: control}
For finite $X,Y\subseteq\atoms$, put
\[
 R_Y=\{(a,B)\in\Theta_p:a\in Y\},\qquad
 C_X=\{(a,B)\in\Theta_p:X\subseteq B\}.
\]
A pair $(X,Y)$ is a \emph{control} of $A\subseteq\Theta_p$ if
$A\subseteq R_Y\cup C_X$. Its \emph{reduction} is
\[
 A/(X,Y)=\{(a,B\setminus X):(a,B)\in A,\ a\notin Y\}.
\]
A \emph{row witness} for $y$ is a pair $(y,B)\in A$, and a
\emph{column witness} for $X$ is a pair $(a,X)\in A$.

\end{definition}

\begin{example}
Let $p=2$ and let $a,b,c,d,e$ be distinct. For
\[
 A=\{(a,\{c,d\}),(e,\{b,c\})\},\qquad X=\{b\},\quad Y=\{a\},
\]
the first pair witnesses row $a$, and the second has a column containing $X$.
Thus $(X,Y)$ controls $A$, with reduction
$A/(X,Y)=\{(e,\{c\})\}$.

\end{example}

\begin{definition}\label{d:short-control}
A set $A\subseteq\Theta_p$ has \emph{good control} $(X,Y)$ of the
indicated type if every pair is allowed and the witness condition holds.
Row witnesses are required for every $y\in Y$ wherever indicated.
Require $A\ne\emptyset$ for non-null types and $|X|=p$ for
row-, column-, and full control.

\[
\begin{array}{c|c|l}
 \text{type} &\text{allowed pairs}&\text{witness condition}\\ \hline
 \nullcontrol &\emptyset & A=\emptyset\\
 \rowcontrol &R_Y & \text{row witnesses}\\
 \columncontrol&C_X & \text{a column witness}\\
 \fullcontrol &R_Y\cup C_X &\text{row witnesses and a column witness}\\
 \unrestrainedcontrol&R_Y\cup C_X &|X|<p,\ \text{row witnesses, and }A/(X,Y)\text{ unrestrained}
\end{array}
\]
The abbreviations $\rowcontrol,\columncontrol,\fullcontrol,\unrestrainedcontrol$
denote row-, column-, full-, and unrestrained-control, respectively.
\end{definition}

\begin{example}\label{e: good control types}
Let $p=1$ and let $a,b,c,d,e$ be distinct atoms.
\begin{enumerate}
\item \emph{Row-control.}
Let $A=\{(a,\{b\}),(a,\{c\}),(a,\{d\}),(a,\{e\})\}.$
A good control is $(X,Y)=(\{b\},\{a\})$.
Every pair lies in row $a$, and $(a,\{c\})$ witnesses that row.

\item \emph{Column-control.}
Let
$A=\{(a,\{b\}),(c,\{b\}),(d,\{b\}),(e,\{b\})\} $
A good control is $(X,Y)=(\{b\},\emptyset)$.
Every pair has column $X$, and $(c,\{b\})$ witnesses that column.

\item \emph{Full control.} Let
$A=\{(a,\{c\}),(a,\{d\}),(a,\{b\}),(c,\{b\}),(e,\{b\})\}. $
A good control is $(X,Y)=(\{b\},\{a\})$.
The first three pairs lie in row $a$, and the last three have column
$\{b\}$. Thus $(a,\{b\})$ witnesses both the row and the column.

\item \emph{Unrestrained control.} Take the set $A$
from \Cref{e: ur set ex}.
With $X=Y=\emptyset$, every pair is allowed
and the reduction is $A$ itself.
Thus $(\emptyset,\emptyset)$ is a good
control of unrestrained type.
\end{enumerate}
\end{example}

\begin{example}
Let $p=2$ and let $a,b,c,d,e,f$ be distinct atoms.
Consider
\[
 A=\{ (f,\{a,b\}),(f,\{c,d\}),
 (a,\{e,b\}),(b,\{e,c\}),(c,\{e,a\})\}.
\]
This set has good control $(X,Y)=(\{e\},\{f\})$ of unrestrained type.
The first two pairs witness row $f$, and every other column contains
$e$. Deleting row $f$ and removing $e$ from the remaining columns gives
\[
 A/(X,Y)=\{(a,\{b\}),(b,\{c\}),(c,\{a\})\}.
\]
This is the unrestrained set in \Cref{e: ur set ex}, and
$|X|=1<p$.
\end{example}

\begin{lemma}\label{l: good control}
Every $A\subseteq\Theta_p$ has a good control $(X,Y)$ with
$|X|\le p$ and $|Y|\le y_p=\binom{p+1}{2}$.
\end{lemma}
\begin{proof}
For $A=\emptyset$, use null control. Otherwise start with $X=Y=\emptyset$,
retaining one witness for each row added to $Y$. Maintain the invariant
that every pair outside those rows has column containing $X$.
\emph{Base.} The invariant holds for $X=Y=\emptyset$.

\emph{Step.} Put $q=p-|X|$ and $R=A/(X,Y)$.
If $R=\emptyset$, pad $X$ to size $p$ and stop with row-control.
If $q=0$, stop with full control. If $R$ is unrestrained, stop with
unrestrained-control. Otherwise choose a restraint $(z,W)\in\Theta_q$
over $\atoms\setminus X$, add the rows of $R$ in $W$ to $Y$, and adjoin
$z$ to $X$. Every remaining reduced column contains $z$, proving the invariant.

Each continuing step increases $|X|$ and adds at most $p-|X|$ rows
using its old value. Thus there are at most $p$ steps and
$|Y|\le\sum_{i=0}^{p-1}(p-i)=y_p$.
At termination the retained row witnesses suffice, with any remaining
pair supplying the column witness when $q=0$.
\end{proof}

\begin{lemma}\label{l: control witnesses}
If $A\subseteq\Theta_p$ has good control $(X,Y)$, some $A_0\subseteq A$
with $|A_0|\le6p^3+|Y|$ has the following property:
every set of allowed pairs containing $A_0$ has the same good control.
\end{lemma}
\begin{proof}
Keep the required row and column witnesses. For unrestrained-control,
also lift an unrestrained subset of the reduction by \Cref{l: unrest simple}.
With $q=p-|X|\ge1$, this adds at most $b_q\le6q^3\le6p^3$ pairs.
Thus at most $|Y|+6p^3$ pairs suffice, including the empty choice for
null control. Any allowed superset retains the witnesses and, by
\Cref{r: rest}, the required unrestrained reduction.

\end{proof}

\begin{lemma}\label{l: control insertion}
Let $A\subseteq\Theta_p$ have good control $(X,Y)$.
Let $(c,D)$ be an allowed pair.
Then there is a pair $(a,B)\in A$
that satisfies:
\begin{enumerate}
\item $(a,D),(c,B)$ are also allowed pairs.
\item $A^\prime$ has the same
good control of the same type, where
\[
A^\prime = (A\setminus\{(a,B)\})\cup\{(a,D),(c,B)\}\, .
\]
\item $A^\prime\cup\{(a,B)\}$
has also the same good control of the same type.
\end{enumerate}
\end{lemma}
\begin{proof}
If row witnesses are required and $c\in Y$, cross with its row witness.
For column-control, or full control with $c\notin Y$, one has $D=X$,
so use the column witness.
Both operations merely add $(c,D)$.
Otherwise the type is unrestrained, $c\notin Y$, and $X\subseteq D$.
Apply \Cref{l: unrest crossing} to $(c,D\setminus X)$ in the reduction,
then restore $X$ to both columns. Its rows avoid $X$, so the crosses
are legal. Row witnesses survive and the reduction stays unrestrained.
Null control has no allowed pair.
\end{proof}

\begin{lemma}\label{l: control crossing}
Fix a good control $(X,Y)$. In every sequence of more than
$(2p+1)(|Y|+1)$ allowed pairs, two positions $i<j$ can be crossed
so that both resulting pairs are legal and allowed.
\end{lemma}
\begin{proof}
Group the sequence by its rows in $Y$, with one further group for the
remaining pairs. For column-control use one group. Some group has
$m\ge2p+2$ positions. Equal rows cross legally. If all rows are distinct,
each column forbids at most $p$ of them, while
\[
 mp<\binom m2.
\]
Hence an unordered pair of positions permits both crosses. In a
witnessed-row group the new rows stay in $Y$. In the remaining group
both columns still contain $X$. Thus both new pairs are allowed.
\end{proof}

\subsection{Cutting and inserting bounded remainders}\label{a: as sl}

For an altering-set word $\bsigma$ as in~\eqref{eq:as} and $t,s\in S$, put
\[
 A^{t,s}_\sbsigma=
 \{(a_i,B_i\cup\{a_{i+1}\}):1\le i<n,\ t_i=t,\ s_{i+1}=s\}.
\]

\begin{definition}\label{d: control profile}
A \emph{control profile} $\chi$ assigns to each $(t,s)\in S^2$ a control
type and a pair $(X,Y)$ with $|X|\le p$ and $|Y|\le y_p$.
Let $Q^\chi$ consist of the words $\bsigma\in Q_{c,c'}$ for which
$A^{t,s}_\sbsigma$ has good control $\chi(t,s)$ for every $(t,s)\in S^2$.
\end{definition}

\begin{definition}\label{d: remainder}
A \emph{remainder} for $\chi$ is a word $g\btau$ with
\[
 g=\tuple{t,E,s},\qquad
 \btau\in Q_{\conf{s}{f},\conf{t}{d}},\qquad f\notin E,
\]
whose internal joins are allowed by $\chi$ at their respective state pairs,
and whose \emph{closing join} $(d,E\cup\{f\})\in\Theta_p$ is allowed
by $\chi(t,s)$. Let $R^\chi$ be the language of these remainders.

\end{definition}

For a language $U$, put $U_N=\{u\in U:|u|\le N\}$.
Set
\[
 W=|S|^2(6p^3+y_p),\qquad
 T=|S|^2(2p+1)(y_p+1)+1.
\]
Here $W$ bounds the number of witness occurrences, and $T$ bounds the
number of consecutive joins needed to find a legal crossing. Define
\begin{equation}\label{eq:remainder-bounds}
 N_1=2T,\qquad N_0=2\bigl(W+(W+1)T\bigr)+1.
\end{equation}

\begin{lemma}\label{l:insert-remainder}
For $u\in Q^\chi$ and $r\in R^\chi$, there is $u'\in Q^\chi$ with
\[
 \Par{u'}=\Par u+\Par r.
\]
\end{lemma}
\begin{proof}
Write $r=g\btau$, with closing join $(d,E\cup\{f\})$ at $(t,s)$.
By \Cref{l: control insertion}, choose a join $(a_i,D_i)\in A_u^{t,s}$
that can be crossed with it while preserving good control.
Factor $u=v g_i w$, where $v$ ends in $z_i$, and put
\[
 u'=v g\btau g_i w.
\]
Its new boundary joins are $(a_i,E\cup\{f\})$ and $(d,D_i)$.
The crossing preserves good control, and all other added joins are allowed
by \Cref{d: remainder}. These additions preserve the row and column witnesses.
For unrestrained control the reduction can only grow, so it remains
unrestrained by \Cref{r: rest}.
Thus $u'\in Q^\chi$, and the inserted letters have vector $\Par r$.
\end{proof}

\begin{lemma}\label{l:unmarked-window}
Let $u=z_1g_1z_2g_2\cdots z_{n-1}g_{n-1}z_n\in Q^\chi$ satisfy
$|u|>N_0$. For every $M\subseteq\{1,2,\ldots,n-1\}$ with $|M|\le W$,
there are $i<j$ at the same state pair such that
\[
 j-i<T,\qquad \{i,i+1,\ldots,j\}\cap M=\emptyset,
\]
and the crossed joins $(a_i,D_j),(a_j,D_i)$ are legal and allowed.
\end{lemma}
\begin{proof}
Put $J=n-1$. Since $|u|=2J+1$, \eqref{eq:remainder-bounds} gives
$J>W+(W+1)T$. Removing the positions in $M$ leaves at most $W+1$
intervals containing more than $(W+1)T$ positions altogether.
One therefore contains $T$ consecutive positions. Among these, some
state pair occurs more than $(2p+1)(y_p+1)$ times, by the definition of $T$.
Apply \Cref{l: control crossing} to its joins. The resulting positions
$i<j$ lie in that interval, so they satisfy both displayed conditions.
\end{proof}

\begin{lemma}\label{l:extract-remainder}
For $u\in Q^\chi$ with $|u|>N_0$, there are
$u'\in Q^\chi$ and $r\in R^\chi_{N_1}$ such that
\[
 |u'|<|u|,\qquad \Par u=\Par{u'}+\Par r.
\]
\end{lemma}
\begin{proof}
At each state pair, mark one occurrence of each witness from
\Cref{l: control witnesses}. The marked positions form a set $M$ with
$|M|\le W$. Choose $i<j$ by \Cref{l:unmarked-window} and put
\begin{equation}\label{eq:deleted-remainder}
 r=g_i z_{i+1}g_{i+1}z_{i+2}\cdots g_{j-1}z_j,\qquad
 u=v r g_j w,\qquad u'=v g_j w.
\end{equation}
The new join $(a_i,D_j)$ and the closing join $(a_j,D_i)$ of $r$ are
legal and allowed at their common state pair. All internal joins of $r$
are allowed, and $0<|r|=2(j-i)\le N_1$, so $r\in R^\chi_{N_1}$.
All remaining joins are allowed, and every marked witness survives.
Thus $u'\in Q^\chi$ by
\Cref{l: control witnesses}. The factorization in~\eqref{eq:deleted-remainder}
proves the vector identity and strict length decrease.
\end{proof}

\subsection{Semi-linearity from bounded words}\label{a: bounded decomposition}

\begin{lemma}\label{l:profile-decomposition}
For every control profile $\chi$,
\begin{equation}\label{eq:q-profile}
 \Par{Q^\chi}=\Par{Q^\chi_{N_0}}+
                 \bigl(\Par{R^\chi_{N_1}}\bigr)^\ast.
\end{equation}
\end{lemma}
\begin{proof}
For $\supseteq$, induct on the number $q$ of remainders.
\emph{Base $q=0$.} The base word lies in $Q^\chi$.
\emph{Step.} Given a word realizing the base vector and the first $q$
remainder vectors, insert the next remainder by \Cref{l:insert-remainder}.

For $\subseteq$, use strong induction on $m=|u|$ for $u\in Q^\chi$.
\emph{Base $m\le N_0$.} Take $u$ as the base and no remainders.
\emph{Step $m>N_0$.} By \Cref{l:extract-remainder}, there are
$u'\in Q^\chi$ and $r\in R^\chi_{N_1}$ with $|u'|<m$ and
$\Par u=\Par{u'}+\Par r$. Apply the induction hypothesis to $u'$
and append $r$, proving~\eqref{eq:q-profile}.
\end{proof}

\begin{proof}[Proof of \Cref{l:LF2}]
For fixed $\chi$, the bounded languages in~\eqref{eq:q-profile} are
finitely supported. Their images are orbit-finite by \Cref{l: b-o-f},
so \Cref{l:profile-decomposition} makes $\Par{Q^\chi}$ semi-linear.
By \Cref{l: good control}, every word admits a profile, hence
\[
 Q_{c,c'}=\bigcup_\chi Q^\chi.
\]
There are finitely many state pairs and control types, and the sets
$X,Y$ have uniform size bounds. Thus the profiles form an orbit-finite
family, also after fixing $c,c'$. The constructions commute with
permutations fixing these configurations, so the displayed union has
a semi-linear Parikh image.
\end{proof}

\section{Rationality of one-register context-free grammars}
\label{s: rat one cfg}\label{a: rest cfg}

Fix a restricted $\CFG_1$ $\bG$ of branching degree $m$.
We follow~\cite[Section~V]{HofmanJLP21}, replacing its one-letter
traversals by block transitions to retain all sibling constraints.

\subsection{Traversals as block transitions}

For a nonterminal production
\[
 \conf A a\to\conf{B_1}{b_1}\conf{B_2}{b_2}\cdots\conf{B_q}{b_q},
\]
an \emph{$i$-traversal} follows child $i$. Its \emph{side-effect} lists the other child configurations
in order, and its \emph{type} is the production rule together with $i$.
Let $S_{c,c'}$ be the sequences of side-effects along paths from $c$ to $c'$.

The following lemma extends~\cite[Lemma~12]{HofmanJLP21} to arbitrary branching.
\begin{lemma}\label{l:grammar-path-images}
Every $S_{c,c'}$ has a rational Parikh image.
\end{lemma}
\begin{proof}
Use a $\BFMA_1$ with states $V$. Each $i$-traversal becomes a transition
from $A$ to $B_i$, reading the other configurations as one block and
updating the register from $a$ to $b_i$.
Its constraint is the full production constraint, with sibling atoms as input
variables. The block bound is at most $\max\{m-1,0\}$, and its language
between $c,c'$ is exactly $S_{c,c'}$. Apply \Cref{l: two_sh}.
\end{proof}

\par\noindent\begin{minipage}{\linewidth}
\begin{example}\label{e: traversal block figure}
In \Cref{fig: traversal block}, the production instance has distinct atoms
$a,b_1,b_2,b_3,b_4$. Its $2$-traversal follows $\conf{B_2}{b_2}$.
The corresponding block transition reads the other three configurations
and updates the register from $a$ to $b_2$, using the displayed constraint.
\begin{figure}[H]
\centering
\begin{tikzpicture}[x=1cm,y=1cm,every node/.style={font=\small},>=stealth]
\node (a) at (2.1,0) {$\conf A a$};
\node (b1) at (0,-1.15) {$\conf{B_1}{b_1}$};
\node[draw=black,rounded corners=2pt,fill=white,inner sep=4pt]
 (b2) at (1.4,-1.15) {$\conf{B_2}{b_2}$};
\node (b3) at (2.8,-1.15) {$\conf{B_3}{b_3}$};
\node (b4) at (4.2,-1.15) {$\conf{B_4}{b_4}$};
\draw[black] (a)--(b1) (a)--(b3) (a)--(b4);
\draw[black,very thick] (a)--(b2);
\node[font=\scriptsize] at (2.1,.55) {production instance};
\node[font=\scriptsize] at (8.35,.55) {one block transition};
\node (p) at (5.85,-.5) {$\conf A a$};
\node (q) at (10.85,-.5) {$\conf{B_2}{b_2}$};
\draw[->,black,thick] (p)--(q);
\node[above,inner sep=4pt] at (8.35,-.5)
 {$\conf{B_1}{b_1}\,\conf{B_3}{b_3}\,\conf{B_4}{b_4}$};
\node at (8.35,-1.15) {$\neq(a,b_1,b_2,b_3,b_4)$};
\end{tikzpicture}
\caption{A traversal and its side-effect block.}
\label{fig: traversal block}
\end{figure}
\end{example}
\end{minipage}
\par

\subsection{Compatibility}

\begin{definition}\label{d: compatible tuples}
For a tuple $\alpha=(a_1,a_2,\ldots,a_k)$ and an atom $c$, put
$I_\alpha(c)=\{i:a_i=c\}$. Two tuples $\alpha,\beta$ are
\emph{compatible}\footnote{This definition coincides
with compatibility as defined in~\cite{HofmanJLP21}.} if they have the same equality type and
$I_\alpha(c)=I_\beta(c)$ for every atom $c$ occurring in both.
For $I\subseteq\{1,2,\ldots,k\}$, write $\alpha_I$ for the tuple
of entries at positions in $I$, in increasing order.
\end{definition}

\begin{lemma}[{\cite[Lemma~11]{HofmanJLP21}}]\label{lem:Ramsey}
There is a nondecreasing function $f$ such that every multiset of at
least $f(k)$ occurrences of $k$-tuples contains two compatible occurrences.
\end{lemma}
\begin{proof}
First fix the equality type. We show that $g(0)=2$ and
$g(k)=1+k^2g(k-1)$ suffice. For $k=0$ all tuples are identical.
For $k>0$, fix one occurrence $\alpha$. Any tuple disjoint from
$\alpha$ is compatible with it. Otherwise each remaining occurrence
shares one of its entries with one of the $k$ entries of $\alpha$.
At least $g(k-1)$ therefore have the same atom $c$ at the same position.
Delete that position and apply the induction hypothesis to their common
restricted type. Restoring $c$ preserves compatibility: its positions
were identical in all selected tuples by their original equality type.
There are finitely many equality types at each arity. Multiplying $g(k)$
by their number and taking the maximum up to $k$ gives the stated $f$.
\end{proof}

\begin{lemma}\label{l: comp rest}
If $\alpha,\beta$ are compatible, so are $\alpha_I,\beta_I$ for every $I$.
\end{lemma}
\begin{proof}
Restriction preserves equality type. For each shared atom $c$, it retains
the same positions $I\cap I_\alpha(c)=I\cap I_\beta(c)$ in both tuples.
\end{proof}

\begin{lemma}[{\cite[Claim~3 in the proof of Lemma~13]{HofmanJLP21}}]
\label{l: comp two pairs}
Suppose $\alpha,\beta$ are compatible $k$-tuples, as are
$\alpha',\beta'$. If the concatenations $(\alpha,\beta)$ and
$(\alpha',\beta')$ are compatible, then all four tuples are pairwise compatible.
\end{lemma}
\begin{proof}
Restricting to either half gives compatibility of $\alpha,\alpha'$
and of $\beta,\beta'$ by \Cref{l: comp rest}. Thus all four have the
same equality type. If $c$ occurs in $\alpha$ and $\beta'$, it occurs
in both concatenations. Comparing their first halves gives
$I_\alpha(c)=I_{\alpha'}(c)\ne\emptyset$. Compatibility within the
second pair then gives $I_{\alpha'}(c)=I_{\beta'}(c)$.
Interchanging the two pairs proves compatibility of $\alpha'$ and $\beta$.
\end{proof}

\subsection{Bounding widths and heights}

The \emph{length} of a path counts traversals whose parent and child atoms differ.
The \emph{height} of a node is the greatest length of a path from it to a leaf.
The \emph{width} of a tree is the largest $n$ for which it has $n$ pairwise
incomparable nodes of height at least $n$.
Let $\langh{c,j}$ and $W_{c,n}$ be the yields of complete trees rooted at $c$
of height at most $j$ and width at most $n$, respectively.

Define the following constants
\[
t=\max\{1,m|P|\} \qquad \mbox{and} \qquad
n=t f(2m+2)
\]
\begin{lemma}\label{l:wide-tuples}
Every tree of width at least $n$ contains two disjoint paths, each
with two updating traversals of one common type, whose four tuples of
parent and child atoms are pairwise compatible.
\end{lemma}
\begin{proof}
Choose $n$ incomparable nodes of height at least $n$ and, below each,
a path with at least $n$ updates. There are at most $t$ traversal types.
On each path some type occurs at least $f(2m+2)$ times. Its tuples
have a common arity at most $m+1$, so \Cref{lem:Ramsey} gives two
compatible occurrences, in path order. At least $f(2m+2)$ paths chose
the same type. Their concatenated pairs have a common arity at most
$2m+2$. Apply \Cref{lem:Ramsey} again and then \Cref{l: comp two pairs}.
\end{proof}

\begin{lemma}[{cf.~\cite[Lemma~13]{HofmanJLP21}}]\label{l:width-arbitrary}
For some $n\ge1$, every complete tree has a tree of width at most $n$
with the same root and Parikh image of its yield.
\end{lemma}
\begin{proof}
Fix the root, node count and yield vector. Among the corresponding trees,
choose one whose decreasing list of root-to-leaf path lengths is
lexicographically greatest. Only finitely many lists are possible.
If its width is at least $n$ from \Cref{l:wide-tuples}, choose the two
paths given there. Let $u,w$ be their upper selected children, ordered
so the greatest root-to-leaf length through $u$ is at most that through
$w$, say $R$. Let $v$ be the lower selected child below $u$.

Move the subtree at $u$ with the subtree at $v$ removed to the position
of $w$, fill its hole with the subtree at $w$, and place the subtree at
$v$ in the former position of $u$. All three attachments use the same
child position of the same rule. A replacement atom either agrees at
that position or differs from the parent and every sibling, by tuple
compatibility and $\mathsf{AllDiff}$. Thus all three updates remain
valid, and the nodes and terminal multiset are preserved.

Every path formerly through $w$ gains the length from $u$ to $v$,
which includes the lower update. In particular a path of length $R$
becomes longer. All other affected paths formerly went through $u$
and had length at most $R$. The sorted list therefore increases,
contradicting its choice.
\end{proof}

\begin{lemma}[{cf.~\cite[Lemma~14]{HofmanJLP21}}]\label{l:height-arbitrary}
For every $j\ge0$ and $c$, $\langh{c,j}$ has a rational Parikh image.
\end{lemma}
\begin{proof}
Induct on $j$, for all roots $c=\conf A a$.
\emph{Base $j=0$.} At fixed $a$, preserving productions form a finite grammar
with terminal rules replaced by markers for their orbit-finite outputs.
Parikh's theorem and \Cref{l:subst} give rationality.
\emph{Step.} Replace each updating rule $e$ by a single terminal marker
and substitute for it the language
\[
 \bigcup_{\varphi_e(a,b_1,b_2,\ldots,b_q)}
 \langh{\conf{B_1}{b_1},j}\,\langh{\conf{B_2}{b_2},j}\cdots\langh{\conf{B_q}{b_q},j}.
\]
Here $\varphi_e$ is the full rule constraint. The induction hypothesis makes
this orbit-finite union rational in Parikh image. Cutting at the first update
on each path, or substituting there, gives exactly $\langh{c,j+1}$.
Apply Parikh's theorem and \Cref{l:subst} again.
\end{proof}

Let $S^n_{c,c'}$ be $S_{c,c'}$ with $\langh{d,n}$ substituted for every
configuration $d$. Its Parikh image is rational by
\Cref{l:grammar-path-images,l:height-arbitrary,l:subst}.

For a complete tree $\T$ of height greater than $n$, let $H_\T$ consist
of nodes of height greater than $n$, and $D_\T$ of its leaves and nodes
with at least two children in $H_\T$.
Let $\mathcal S_\T$ record endpoint configurations of paths from children
in $H_\T$ of nodes in $D_\T$ to the next node in $D_\T$, including the
initial root path. Let $\mathcal L_\T$ record configurations of children
outside $H_\T$ of nodes in $D_\T$. Both records are multisets.

\begin{lemma}[{cf.~\cite[proof of Lemma~15]{HofmanJLP21}}]\label{l:grammar-skeleton}
For every $c$ and $n\ge1$, there is a rational set $X_{c,n}$ with
$\Par{W_{c,n}}\subseteq X_{c,n}\subseteq\Par{\langpar c\bG}$.
\end{lemma}
\begin{proof}
For width at most $n$, $H_\T$ has at most $n$ leaves and $n-1$ branching nodes.
Thus $|\mathcal S_\T|\le2n-1$ and $|\mathcal L_\T|\le m(2n-1)$.
The family $\mathcal F_{c,n}$ of records from such trees rooted at $c$
is therefore orbit-finite. Put
\begin{equation}\label{eq:grammar-total-image}
 X_{c,n}=\Par{\langh{c,n}}\ \cup\!
 \bigcup_{(\mathcal L,\mathcal S)\in\mathcal F_{c,n}}
 \left(\sum_{d\in\mathcal L}\Par{\langh{d,n}}
 +\sum_{(e,e')\in\mathcal S}\Par{S^n_{e,e'}}\right).
\end{equation}
Cutting a tree at $D_\T$ gives the first inclusion.
Conversely, retain the productions at $D_\T$ in a tree witnessing a record
and replace its recorded paths and remaining child trees by the chosen factors.
This preserves their endpoint configurations and gives the second inclusion.
All factors are rational. Choose their expressions equivariantly as in
\Cref{l:subst}, fixing $c$ and the distinguished terminals, so
\eqref{eq:grammar-total-image} is rational.
\end{proof}

\begin{proof}[Proof of \Cref{t: rat one cfg}]
Choose $n$ from \Cref{l:width-arbitrary} and let $c_0$ be the initial configuration.
Then $\Par{\langof\bG}=\Par{W_{c_0,n}}=X_{c_0,n}$ by
\Cref{l:grammar-skeleton}, hence is rational.
Apply \Cref{l: par res cfg} to remove the restriction on $\bG$.
\end{proof}

\section{Linear forms of sets of data vectors}
\label{s: lin form}

\begin{definition}
For star-height $h=0$,
a rational set is necessarily orbit-finite.
Such a set is said to be in linear form of height $0$
if it is presented as an orbit-finite union of singletons:
\begin{equation}
\label{eq: sh zero}
\bigcup_{i\in I}g_i\,,
\end{equation}
where $g_i$ is a data vector for $i\in I$.

For star-height $h>0$,
a rational set is in linear form
if it is presented as an orbit-finite union
of singletons and lower-height linear forms, namely:
\begin{equation}
\label{eq: normal ratset}
\bigcup_{i\in I} \, g_i + {P_i}^\ast\,,
\end{equation}
where $g_i$ is a data vector and $P_i$
is a rational set of data vectors
that is already in
a linear form of star-height strictly smaller than $h$,
for all $i\in I$. We call $g_i$ the \emph{base} and the specified expression
for $P_i$ the \emph{period form}.
\end{definition}

\begin{example}
\label{e: dv set}
Let
\begin{align*}
P_a&=\bigcup_{b\in\atoms\setminus\set{a}}a+2 b, \quad
\mbox{for each }a\in \atoms\, ,\\
R&=\bigcup_{a\in\atoms}a+\left(P_a\right)^\ast \, .
\end{align*}
Intuitively, $R$ consists of all the data vectors
which for some atom $a$,
contain $a$ exactly $k+1$ times
and contain $k$ pairs of equal atoms distinct from $a$,
for some $k\geq 0$.
For instance, $3 a+2 b+2 c$
and $3 a+4 b$ belong to $R$.
Furthermore, $R$ is in linear form of height one.
\end{example}

\begin{proposition}[Cf. {\cite[Proposition 9]{HofmanJLP21}}]
\label{p: lin form}
Every rational set of data vectors admits a representation
in linear form of height at most the star-height of any expression defining it.
\end{proposition}

\begin{proof}
Use structural induction on the expression. \emph{Base.} The empty union,
zero, and letters have height-zero forms.
\emph{Step.} Assume forms for all operands. Flatten a union by pairing
its outer and inner indices. Distribute a sum and combine its summands by
\[
 (g+P^\ast)+(g'+{P'}^\ast)=(g+g')+(P\cup P')^\ast,
\]
with empty periods at height zero. The period union has smaller height
than the larger positive-height summand. A star has base $\zerovector$
and the operand as its period form, increasing height by one.
All constructions commute with renaming. The paired indices are orbit-finite:
choose finitely many outer representatives, then finitely many inner
representatives over each representative's support and the common support.
\end{proof}

\begin{definition}
\label{d: pars}
Let $R$ be a rational set of data vectors 
presented in linear form,
and let $v\in R$ be a data vector.
A \emph{parsing tree} of $v$
is a rooted tree whose vertices
are labeled by pairs consisting of a data vector
and a finite subset of atoms.
Parsing trees are defined recursively according to the height of the presentation.

For star-height $h=0$, $R$ is of the form \eqref{eq: sh zero},
therefore, there is $i\in I$ such that $v=g_i$.
Let $S_v=\supp{f}$ where $f$ is the mapping $x\mapsto g_x$,
which is a finitely supported function,
because it is an orbit-finite union.
The parsing tree $T_v$ consists of
a single root vertex
labeled by $(v,S_v)$.

For star-height $h>0$, $R$ is of the form \eqref{eq: normal ratset}.
Then
there are $i\in I,v_0= g_i$
and $v_1,v_2,\ldots,v_k\in P_i$,
such that $v=\sum_{j=0}^kv_j$.
The root of the parsing tree $T_v$ is labeled by $(v_0,\supp{f})$
(where $f$ maps $x$ to $g_x$ and the chosen indexed linear form for $P_x$),
and it has $k$ children whose sub-trees
are the parsing trees
$\{T_{v_j}\}_{j=1}^k$
of the corresponding data vectors in $P_i$.
\end{definition}

Note that the depth of any parsing tree $T_v$
is at most the height of the chosen linear form for $R$.

Given a parsing tree $T$ and a sub-tree $T^\prime$ of $T$
(i.e., a node together with all its descendants),
we define the \emph{value} of $T^\prime$
as the sum of all data vectors occurring in its vertex labels.
In particular, the value of a parsing tree
$T_v$ is the vector $v$ itself.

\begin{example}
\label{e: dv set 2}
Continuing \Cref{e: dv set},
the data vectors
$3 a+2 b+2 c$
and $3 a+4 b$
both admit parsing trees of depth one,
as imposed by the structure of $R$.
Their parsing trees are depicted in \Cref{fig: pars trees}.
\end{example}

\begin{figure}[htbp]
\centering

\begin{subfigure}{0.45\textwidth}
\centering
\begin{tikzpicture}[
level distance=22mm,
sibling distance=25mm,
every node/.style={draw, rounded corners, align=center, inner sep=3pt}
]

\node {$a,\emptyset$}
child { node {$a+2 b,\set{a}$} }
child { node {$a+2 c,\set{a}$} };
\end{tikzpicture}
\caption{Parsing tree for $3 a+2 b+2 c$}
\end{subfigure}
\hfill
\begin{subfigure}{0.45\textwidth}
\centering
\begin{tikzpicture}[
level distance=22mm,
sibling distance=25mm,
every node/.style={draw, rounded corners, align=center, inner sep=3pt}
]

\node {$a,\emptyset$}
child { node {$a+2 b,\set{a}$} }
child { node {$a+2 b,\set{a}$} };
\end{tikzpicture}
\caption{Parsing tree for $3 a+4 b$}
\end{subfigure}

\caption{Parsing trees for data vectors in Example~\ref{e: dv set}}
\label{fig: pars trees}
\end{figure}

The following operations use the fact that each child is one summand in a star.

\begin{lemma}\label{p: tree op}
Parsing trees for $R$ are closed under
\begin{enumerate}
\item deletion of proper subtrees,
\item duplication of proper subtrees at the same parent,
\item permutation of a subtree fixing pointwise the atom-set in its root label.
\end{enumerate}
\end{lemma}
\begin{proof}
Deleting or duplicating a child removes or repeats a summand of its parent's star. A permutation fixing the subtree's root support preserves its indexed presentation, so the renamed choices also form a parse.
\end{proof}

\begin{example}
Continuing Examples~\ref{e: dv set} and~\ref{e: dv set 2},
\Cref{fig: pars op} depicts possible operations
on the parsing tree of $3 a+2 b+2 c$
on the sub-tree of the vertex labeled $(a+2 c,\set{a})$,
which creates parsing trees for the data vectors,
$2 a+2 b$, $4 a+2 b+4 c$, and
$3 a+2 b+2 d$.
\end{example}

\begin{figure}[htbp]
\centering

\begin{subfigure}{0.22\textwidth}
\centering
\begin{tikzpicture}[
scale =0.8,
level distance=22mm,
sibling distance=30mm,
every node/.style={draw, rounded corners, align=center, inner sep=3pt}
]

\node {$a,\emptyset$}
child { node {$a+2 b,\set{a}$} };

\end{tikzpicture}
\caption{Pruning a sub-tree}
\end{subfigure}
\hfill\begin{subfigure}{0.415\textwidth}
\centering
\resizebox{\linewidth}{!}{\begin{tikzpicture}[
scale =0.8,
level distance=22mm,
sibling distance=26mm,
every node/.style={draw, rounded corners, align=center, inner sep=3pt}
]

\node {$a,\emptyset$}
child { node {$a+2 b,\set{a}$} }
child { node {$a+2 c,\set{a}$} }
child { node {$a+2 c,\set{a}$} };

\end{tikzpicture}}
\caption{Duplicating a sub-tree}
\end{subfigure}
\hfill
\begin{subfigure}{0.32\textwidth}
\centering
\begin{tikzpicture}[
scale =0.8,
level distance=22mm,
sibling distance=26mm,
every node/.style={draw, rounded corners, align=center, inner sep=3pt}
]

\node {$a,\emptyset$}
child { node {$a+2 b,\set{a}$} }
child { node {$a+2 d,\set{a}$} };

\end{tikzpicture}
\caption{Permuting a sub-tree ($c\mapsto d$)}
\end{subfigure}

\caption{Parsing tree operations
applied to a parsing tree from~\Cref{e: dv set}.}
\label{fig: pars op}
\end{figure}

\begin{proposition}
\label{p: const pars}
Let $R$ be a rational set of data vectors in
linear form.
There are constants $N$ and $M$,
such that, for every $v\in R$
and every parsing tree $T_v$ of $v$,
every label $(u,A)$ in $T_v$ satisfies
$\size{u}\leq N $ and $\card{A} \, \leq M $.
\end{proposition}

\begin{proof}
Induct on the presentation height $h$, with indexing function $f$.
\emph{Base $h=0$.} Take $N=\max_i\size{g_i}$ and $M=|\supp f|$,
with empty maximum zero. Base sizes are constant on the finitely many
index orbits over $\supp f$.
\emph{Step $h>0$.} Assume bounds for every smaller-height period form.
Choose bounds $N_i,M_i$ for one representative per index orbit and transport them.
Then the finite bounds
\[
 N=\max_i\{\size{g_i},N_i\},\qquad
 M=\max\bigl(\{|\supp f|\}\cup\{M_i:i\in I\}\bigr)
\]
cover the root label and, by the induction hypothesis, every descendant label.
\end{proof}

\section{One-register context-free grammars
have large star-height}
\label{s: cfg high sh}

Fix $h\ge1$ and $a\in\atoms$. Let $d_0,u_1,d_1,u_2,\ldots,d_{h-1},u_h$
be distinct labels.
The grammar $\bG_{\mathrm{sh},h}\in\BinaryCFG_1$ has variables $S,S_1,S_2,\ldots,S_{2h-1}$, and production rules
\begin{align*}
 S(x)&\to\varepsilon\mid S_1(x),\\
 S_{2i+1}(x)&\to\varepsilon\mid S_{2i+1}(x)S_{2i+1}(x) &&(0\le i<h),\\
 S_{2i+1}(x)&\to d_i(x)S_{2i+2}(y)\quad(x\ne y) &&(0\le i<h-1),\\
 S_{2i+2}(y)&\to S_{2i+3}(y)u_{i+1}(y) &&(0\le i<h-1),\\
 S_{2h-1}(x)&\to d_{h-1}(x)u_h(y)\quad(x\ne y).
\end{align*}

\begin{remark}\label{r:height-grammar}
Splitting and empty productions allow any finite list of factors at
$S_{2i+1}(x)$. Each factor pairs $d_i(x)$ with $u_{i+1}(y)$, where
$x\ne y$, around a yield from $S_{2i+3}(y)$ if $i<h-1$.
For $i>0$, every $S_{2i+1}(b)$ not produced by splitting is introduced
together with $u_i(b)$.
\end{remark}

\begin{example}
For $h=3$, the commutative image of $\bG_{\mathrm{sh},3}$ is
\begin{equation}\label{eq: par bg sh 3}
 \left(\bigcup_{b\ne a}\left(
 \letter{d_0,a}+\letter{u_1,b}+
 \left(\bigcup_{c\ne b}\left(
 \letter{d_1,b}+\letter{u_2,c}+
 \left(\bigcup_{d\ne c}
 (\letter{d_2,c}+\letter{u_3,d})\right)^\ast
 \right)\right)^\ast
 \right)\right)^\ast.
\end{equation}
Each factor contributes two letters and may contain a list of factors
from the next level. The three nested stars correspond to the lists
generated at $S_1$, $S_3$, and $S_5$.
\end{example}

For $x\in\atoms$, put $R_{h,x}=\{\zerovector\}$ and, for
$i=h-1,h-2,\ldots,0$, define
\begin{equation}\label{eq:height-family}
 R_{i,x}=\left(\bigcup_{y\ne x}
 \bigl(\letter{d_i,x}+\letter{u_{i+1},y}+R_{i+1,y}\bigr)\right)^\ast.
\end{equation}

\begin{lemma}\label{l:height-family-up}
For $0\le i<h$ and $x\in\atoms$, the commutative image of the language
generated from $\conf{S_{2i+1}}x$ is $R_{i,x}$ and has star-height at most $h-i$.
In particular, $\Par{\langof{\bG_{\mathrm{sh},h}}}=R_{0,a}$ has star-height at most $h$.
\end{lemma}
\begin{proof}
A factor at $S_{2i+1}(x)$ is a yield whose first derivation step emits $d_i$.
Contract the maximal subtrees starting with such a step. Splitting and
empty productions leave a binary tree with one factor or an empty yield at each leaf.
Conversely, every finite list of such factors can be placed at the leaves
of such a tree. The language is therefore the star of the factor language.

Induct downward on $i$, uniformly in $x$.
\emph{Base $i=h-1$.} Each factor has vector
$\letter{d_{h-1},x}+\letter{u_h,y}$ with $y\ne x$, giving $R_{h-1,x}$
of star-height at most one.
\emph{Step $i<h-1$.} Assume the assertion at $i+1$ for every $y$.
A factor has vector in $\letter{d_i,x}+\letter{u_{i+1},y}+R_{i+1,y}$,
and every such vector occurs. Thus~\eqref{eq:height-family} gives $R_{i,x}$.
The orbit-finite union adds no star, and the outer star gives the bound $h-i$.
The start rule adds only $\zerovector$, already in $R_{0,a}$.
\end{proof}

For a label $t$, write $|v|_t=\sum_{b\in\atoms}v(\letter{t,b})$.

\begin{lemma}\label{l:height-balance}
For every $v\in\Par{\langof{\bG_{\mathrm{sh},h}}}$, $b\in\atoms$,
and $1\le i<h$,
\[
 v(\letter{d_i,b})>0\ \Longrightarrow\ v(\letter{u_i,b})>0,
 \qquad |v|_{d_i}=|v|_{u_{i+1}}.
\]
\end{lemma}
\begin{proof}
Choose a derivation with vector $v$.
By \Cref{r:height-grammar}, each factor pairs one $d_i$ with one
$u_{i+1}$, giving the equality. If $d_i(b)$ occurs, follow its
$S_{2i+1}(b)$ ancestors past all splitting productions. The first
other production introduces this variable together with $u_i(b)$,
by \Cref{r:height-grammar}. Hence $v(\letter{u_i,b})>0$.
\end{proof}

\begin{lemma}\label{l:height-witnesses}
For every $m\ge1$, there is $v_m\in\Par{\langof{\bG_{\mathrm{sh},h}}}$ such that
\begin{align*}
 |v_m|_{u_1}&=m,\\
 v_m(\letter{u_i,b})&\in\{0,1\} &&(1\le i\le h,\ b\in\atoms),\\
 v_m(\letter{u_i,b})=1&\ \Longrightarrow\ v_m(\letter{d_i,b})=m
 &&(1\le i<h,\ b\in\atoms).
\end{align*}
\end{lemma}
\begin{proof}
Let $I=\{1,2,\ldots,m\}$ and choose pairwise distinct atoms
$b_s$ for $s\in\bigcup_{i=0}^h I^i$, with $b_\varepsilon=a$.
Put $W_s=\zerovector$ for $|s|=h$ and, for $|s|<h$, set
\[
 W_s=\sum_{q\in I}\bigl(\letter{d_{|s|},b_s}
       +\letter{u_{|s|+1},b_{sq}}+W_{sq}\bigr).
\]
Induct downward on $|s|<h$ to realize $W_s$ from $\conf{S_{2|s|+1}}{b_s}$.
\emph{Base $|s|=h-1$.} Concatenate the $m$ terminal factors in
\Cref{r:height-grammar}.
\emph{Step $|s|<h-1$.} Assume each $W_{sq}$ is realizable.
Use its derivation inside the factor choosing $b_{sq}$ and concatenate
the $m$ factors, again by \Cref{r:height-grammar}.
Then $v_m=W_\varepsilon$ belongs to the image. Distinctness of the $b_s$
gives one incoming letter per nonroot node and $m$ outgoing letters per internal node.
\end{proof}

\begin{corollary}\label{c:marker-propagation}
Let $T$ be a parsing tree of $v\in\Par{\langof{\bG_{\mathrm{sh},h}}}$
in a linear form for this image, and let $w$ be the value of a proper subtree.
For $1\le i<h$ and $b\in\atoms$, if
$w(\letter{u_i,b})=v(\letter{u_i,b})=1$, then
\[
 w(\letter{d_i,b})=v(\letter{d_i,b}),
 \qquad |w|_{d_i}=|w|_{u_{i+1}}.
\]
\end{corollary}
\begin{proof}
Pruning gives $v-w\in\Par{\langof{\bG_{\mathrm{sh},h}}}$ by \Cref{p: tree op}.
Since $(v-w)(\letter{u_i,b})=0$, \Cref{l:height-balance} gives
$(v-w)(\letter{d_i,b})=0$.
Subtracting the balance equations for $v$ and $v-w$ gives the second equality.
\end{proof}

\begin{lemma}\label{l:marker-descent}
The commutative image of $\bG_{\mathrm{sh},h}$ has star-height at least $h$.
\end{lemma}
\begin{proof}
Fix a minimum-height linear form for the image, using \Cref{p: lin form}.
Choose a base-size bound $N$ from \Cref{p: const pars} and a parse of
$v_m$ from \Cref{l:height-witnesses}, with $m>N$.
Induct on $i$ to find a depth-$i$ subtree containing a $u_i$-letter.
\emph{Base $i=1$.} The root base holds fewer than $m$ letters,
so a child contains one of the $m$ distinct $u_1$-letters.
\emph{Step $i<h$.} Assume such a subtree of value $w$ contains
$\letter{u_i,b}$. By \Cref{c:marker-propagation,l:height-witnesses},
\[
 |w|_{u_{i+1}}=|w|_{d_i}\ge w(\letter{d_i,b})=m.
\]
Its base holds fewer than $m$ letters, so a child contains a $u_{i+1}$-letter.
The parse has depth at least $h$, proving the bound.
\end{proof}

\begin{proof}[Proof of \Cref{t: high sh cfg}]
Together, \Cref{l:height-family-up,l:marker-descent} give star-height exactly
$h$ for every $h\ge1$. Height zero is realized by the binary one-register
grammar $S(x)\to T(x)T(x)$, $T(x)\to\varepsilon$, whose image is
$\{\zerovector\}$.
\end{proof}

\section{Hierarchical register automata witness the full star-height hierarchy}\label{s: hra}

\subsection*{Lower bound}

For distinct labels $u_i,d_i$, $1\le i\le r$, define the languages
\begin{align*}
 K_{r,a}&=\{\letter{d_r,a}\}^\ast,\\
 K_{i,a}&=\left(\bigcup_{b\in\atoms}
 \letter{d_i,a}\letter{u_{i+1},b}K_{i+1,b}\right)^\ast
 &&(1\le i<r),\\
 L_r&=\left(\bigcup_{a\in\atoms}\letter{u_1,a}K_{1,a}\right)^\ast.
\end{align*}

\begin{lemma}\label{l:hra-recognition}
For every $r\ge1$, an $\HRA_r$ recognizes $L_r$.
\end{lemma}
\begin{proof}
Start with $k=0$. Store levels $a_1,a_2,\ldots,a_k$, $k\le r$, with distinct
values in a register prefix and their positions in finite control.
On $u_1(b)$ reset to $(b)$. On $d_i$, $i\le k$, test $a_i$ and discard later levels.
If $i<r$, require $u_{i+1}(b)$ next and append $b$.
Preserve a prefix of length $p$ by level-$p$ tests, reusing its values
or writing a new value at level $p+1\le r$.
Accept exactly when no $u_{i+1}$ is pending. These transitions follow the defining recursion for $L_r$.
\end{proof}

\begin{lemma}\label{l:hra-sharp}
For every $r\ge1$, the star-height of $\Par{L_r}$ is exactly $r+1$.
\end{lemma}
\begin{proof}
By \Cref{l:hra-recognition}, the star-height is at most $r+1$.
The implication and balances in \Cref{l:height-balance} hold for $i<r$,
and the implication also holds for $i=r$.
Choose $m$ larger than the base-size bound of a minimum-height linear form
and take $m$ factors per star with fresh incoming atoms, as in \Cref{l:height-witnesses}.
The pruning and depth induction in the proof of \Cref{l:marker-descent}
give a depth-$r$ subtree containing a unique $u_r$-letter.
It contains all $m$ matching $d_r$-letters. One lies below its base,
so the parsing tree has depth at least $r+1$.
\end{proof}

\subsection*{Upper bound}

Fix $\bA\in\HRA_r$ with states $S$, accepting states $T$, alphabet
$\Sigma=H\times\atoms$, and initial configuration $\conf{s_0}{\br_0}$.
Let $a_0$ be the initial value of register one. Following the decomposition
of~\cite[Theorem~14]{LasotaP21}, we separate the letters read
when this register is updated from the words read between updates.

For $s,t\in S$ and $a\in\atoms$, let $F_{s,a,t}$ be the language read
from $s$ to $t$ while register one keeps the value $a$. The other
registers may have any admissible endpoint values. Let $F^0_t$ be the
corresponding language from the prescribed configuration $\conf{s_0}{\br_0}$.

A marker $\tuple{i,s,a,t}$ records the entry state $s$, the value $a$ of
register one, and the exit state $t$. The tag $i$ is $0$ for the initial
marker and $1$ thereafter. Markers are disjoint from $\Sigma$.
The language $Z$ consists of the words formed as follows.
\begin{enumerate}
\item Start with $\tuple{0,s_0,a_0,t}$ for any $t\in S$.
\item After a marker with value $a$ and exit state $t$, append
      $\sigma\tuple{1,s,b,u}$ if $\bA$ has a transition from $t$ to $s$
      reading $\sigma$ that updates register one from $a$ to $b\ne a$.
      The exit state $u\in S$ is arbitrary.
\item Stop at any marker whose exit state belongs to $T$.
\end{enumerate}

Substitute the recorded language for each marker and retain each update letter:
\[
 K_{\tuple{0,s_0,a_0,t}}=F^0_t,\qquad
 K_{\tuple{1,s,a,t}}=F_{s,a,t},\qquad K_\sigma=\{\sigma\}\quad(\sigma\in\Sigma).
\]

\begin{lemma}
\begin{equation}\label{eq:hra-substitution}
 \langof{\bA}=Z(K).
\end{equation}
\end{lemma}
\begin{proof}
Cut an accepting run immediately before and after each update of register one.
Its intervening words belong to the corresponding marker languages, giving
$\langof{\bA}\subseteq Z(K)$. Conversely, choose a witnessing run for
each substituted word. By \Cref{d: hra}, a transition updating register
one has level one and tests only its old value, its new value, and the
input atom. It therefore joins the chosen runs with their actual higher
register values. The first run starts at $\conf{s_0}{\br_0}$ and the last
ends in $T$.
\end{proof}

\begin{lemma}\label{l:hra-marker-image}
$\Par Z$ is semi-linear.
\end{lemma}
\begin{proof}
A restricted $\BFMA_1$ reads each marker at its recorded register value.
Finite control enforces the three clauses defining $Z$. For an update
$a\to b$, split its constraint into equality types. If the input atom
equals $a$ or $b$, read the letter at that value and update on an empty
block immediately afterwards or beforehand, respectively. If it differs
from both, read it on the update itself. Private intermediate states
force each pair of rules to be consecutive.

Between updates the automaton reads at most one marker and the two
adjacent input letters. Thus each preserving language in
\eqref{eq:run-substitution} has an orbit-finite Parikh image by
\Cref{l: b-o-f}. Its star-height is zero. By \Cref{l:LF,l:subst},
the image between fixed endpoints has star-height at most $1+0$.
The orbit-finite union over accepting endpoints gives $\Par Z$.
\end{proof}

\begin{lemma}\label{l:hra-fixed-image}
Let $r>1$. Any upper bound on the star-height of Parikh images of
$\HRA_{r-1}$ languages also holds for $F_{s,a,t}$ and $F^0_t$.
\end{lemma}
\begin{proof}
Fix register one at $a$ and delete it. By injectivity, the remaining
registers range over $\atoms\setminus\{a\}$. On this alphabet, encode
the input using two copies of each label, with letter map
\[
 \rho_a(\letter{(h,0),b})=\letter{h,a},\qquad
 \rho_a(\letter{(h,1),b})=\letter{h,b}.
\]
Substituting $x_1=x'_1=a$ and the decoded input in each constraint
lowers every level $i>1$ to $i-1$ and makes level-one tests
register-independent. Thus $\rho_a$ maps the runs of the resulting
$\HRA_{r-1}$ onto the runs keeping register one at $a$.

The singleton substitution $\rho_a$ preserves the bound by
\Cref{l:subst}. For $F^0_t$, use the prescribed initial valuation
with $a=a_0$. For $F_{s,a,t}$, take the orbit-finite union over
initial valuations fixing register one at $a$, which adds no star.
\end{proof}

\begin{corollary}\label{c: hra upper}
The Parikh image of every $\HRA_r$
language has
star-height at most $r+1$.
\end{corollary}
\begin{proof}
For $r=1$, apply \Cref{t: sh two}, since every $\HRA_1$ is an $\FMA_1$.
For $r>1$, assume the assertion for $r-1$. By \Cref{l:hra-fixed-image},
each marker replacement has Parikh image of star-height at most $r$.
The other replacements are singletons, and the family $K$ is supported by
$\supp{\br_0}$.
Applying \Cref{l:subst} to \eqref{eq:hra-substitution}
with \Cref{l:hra-marker-image} gives the bound $1+r$.
\end{proof}

\section{Irrationality of Parikh images of quasi-regular languages}
\label{s: par fma neq rat}

In this section we construct
a quasi-regular language
whose Parikh image is not rational.
We precede the presentation of
the separating language
with an introduction
to the notion of \emph{commutative stability},
that is the key concept
for the construction
of the separating language.

\subsection{Commutative stability}

Different words can have the same Parikh vector. In the separation proof,
duplicating part of a parsing tree adds a small vector $w$ to the vector $v$
of a witness. We introduce \emph{commutative stability} to express the following
property of a witness. After such an increment, every word of the language
with vector $v+w$ must retain the same sets of letters in its blocks, and the
corresponding counters can only stay fixed or increase. This lets us infer which
letters must be added together from their counts alone.

Fix a language $L$ with specified block presentations. A block $u^n$ repeats
a nonempty word $u$ exactly $n\ge1$ times. We call $n$ its \emph{counter}.

\begin{definition}\label{d:commutative-stability}
Fix a block presentation of $z\in L$, put $v=\Par z$, and let $C>0$.
We call $z$ \emph{commutatively stable up to $C$} for this presentation if
the following holds.

For every data vector $w\le v$ with $\size w<C$ and
every block presentation of a word
$z'\in L$ with $\Par{z'}=v+w$
with the same number of blocks, there is
a bijection from the blocks of $z$ to those of $z'$ under which corresponding
blocks have the same set of letters and each new counter is at least the old one.
\end{definition}

\begin{example}
\label{e: three block}
Consider the  finite-alphabet language,
\[
L_{3\mathrm{Block}}=
\set{(abc)^{n_1}(cde)^{n_2}(efg)^{n_3}:
n_i\geq 1 \mbox{ for }i=1,2,3}\, .
\]
In this language, every word is uniquely determined by
its Parikh vector.
Indeed, the values of $a$ and $b$ are the same and determine
the first exponent $n_1$;
the value of $d$ determines the second exponent;
and the values of $f$ and $g$ are the same and determine
the third exponent $n_3$.

Let
$z=(abc)^{n_1}(cde)^{n_2}(efg)^{n_3}\in L_{3\mathrm{Block}}$
with $v=\Par{z}$, and let $w\leq v$ be such that
there exists 
$z^\prime\in L_{3\mathrm{Block}}$ with $\Par{z^\prime}=v+w$.
Since the values of $a,d$, and $g$
have increased or remained unchanged,
we conclude that
$z^\prime = (abc)^{m_1}(cde)^{m_2}(efg)^{m_3}$
with $m_i\geq n_i$ for $i=1,2,3$.

Thus,
increasing the Parikh vector corresponds exactly
to increasing the block exponents.
In particular, every word of this
language is stable for every $C>0$.
\end{example}

This stability phenomenon fails for blocks of size two.

\begin{example}
\label{e: two block}
Consider the finite-alphabet language,
\[
L_{2\mathrm{Block}}=
\set{(ab)^{n_1}(bc)^{n_2}(cd)^{n_3}:n_i\geq 1 \mbox{ for }i=1,2,3}\, .
\]
Again, every word is uniquely determined by
its Parikh vector:
the values of $a$ and $d$ determine
the first and third exponents,
while the second exponent
is determined by the value of either $b$ or $c$.

However, for
$z=(ab)^{n_1}(bc)^{n_2}(cd)^{n_3}\in L_{2\mathrm{Block}}$
with $n_2\geq 2$,
the word
$z^\prime=(ab)^{n_1+1}(bc)^{n_2-1}(cd)^{n_3+1}$
also belongs to $L_{2\mathrm{Block}}$.
Although $\Par{z^\prime}\geq \Par{z}$,
the exponents structure is broken.
Thus, the Parikh vector may increase
while the middle exponent decreases,
so the block structure is not stable under
commutative perturbations.
\end{example}

For the separation we need stable words with arbitrarily many blocks.
We give
each block an atom of its own, its \emph{anchor}, and choose counters that are
are growing exponentially to differentiate
between distinct blocks.
The anchor counts then identify individual
blocks
. The counts of shared atoms record which blocks meet.

\subsection{The separating language}
Fix a distinguished atom $\#$. Let $L$ be the language of all words of the form
\begin{equation}
\label{eq: word l}
\begin{aligned}
(\tau_1\tau_2)^{n_0}
\#\#
(\tau_2\tau_3\tau_4)^{n_1}
&\#\#
(\tau_4\tau_5\tau_6)^{n_2}
\#\#
(\tau_6\tau_7\tau_8)^{n_3}\cdots \\[1ex]
&\#\#
(\tau_{2k}\tau_{2k+1}\tau_{2k+2})^{n_k}
\#\#
(\tau_{2k+2}\tau_{2k+3})^{n_{k+1}}\, .
\end{aligned}
\end{equation}
where
\begin{enumerate}
\item $k\geq 0$,
\item $n_i\geq 1$, $i=0,1,\ldots,k+1$, and
\item $\tau_i\in\atoms\setminus \{\#\}$, $i=1,2,\ldots,2k+3$.
\end{enumerate}

The first and last blocks are of size two,
and all intermediate blocks are of
size three.
Consecutive blocks share a prescribed atom:
the last atom of one block
is the first atom of the next block.

\begin{lemma}
\label{l: l fma3}
$L$ is recognized by a deterministic $\FMA_3$.
\end{lemma}

\begin{proof}
Read nonempty blocks separated by exactly $\#\#$.
Require a positive power of a pair in the first block and at least one later block. Test the first atom of each later block against the preceding block's last atom. Using its first three positions, test periods two and three in parallel, tracking their length residues in finite control.
At a separator require a whole number of triples in a noninitial block.
At end of input require whole pairs.
Equal stored atoms share a slot, so three registers suffice without guessing.
\end{proof}

\begin{theorem}
\label{t: irrational}
$\Par{L}$ is irrational.
\end{theorem}

\Cref{l: l fma3,t: irrational} imply \Cref{t: par fma neq rat}. For the proof assume that $\Par L$ has a linear form $R$ of height $r$. Take a base bound $N_0\ge1$ from \Cref{p: const pars} and take $M_0$ one larger than its support bound, to allow for the separator.

\subsection{Proof idea}
We prove the following three steps.
\begin{enumerate}
\item For a suitable vector, the separator count and the gaps between counters
recover the path of blocks. A small increment preserving membership in $\Par L$
therefore increases an anchor only together with its neighbors.
\item Pruning and duplicating parsing subtrees then show that a subtree
containing all occurrences of an atom also contains its neighbors.
\item The separator count forces many children at one star. Pruning and
permuting copies of these children, together with the second step, force at
least $M_0+1$ distinct neighboring atoms into the period form's support enlarged
by $\{\#\}$.
\end{enumerate}
That set has size at most $M_0$, contradicting the assumed linear form for $\Par L$.

\subsection{Commutative stability in infinite alphabets}

Consider the sub-language $L^\prime\subset L$
containing words in which all symbols $\tau_i$
are pairwise distinct.
We focus on $\Par{L^\prime}\subseteq\Par{L}$.

\begin{proposition}
\label{p: dom bound}
For every $v\in \Par{L}$,
\[
|\dom{v}|\leq v(\#)+2\, ,
\]
where the equality holds if and only if $v\in\Par{L^\prime}$.
\end{proposition}

\begin{proof}
A witness with parameter $k$ has $2k+2$ separators and at most $2k+3$ non-separator atoms. Hence $\card{\dom v}\le 2k+4=v(\#)+2$, with equality exactly when its $\tau_i$ are pairwise distinct.
\end{proof}

In particular, the number of distinct atoms
is linearly bounded by the number
of appearances of $\#$.

Fix $C\ge1$ and a word $\bsigma\in L'$ of the form~\eqref{eq: word l}
with counters $n_i=C4^{i+1}$ for $0\le i\le k+1$.
Put $v=\Par{\bsigma}$. Each block has one atom of its own, its
\emph{anchor} $\tau_{2i+1}$.
Each even-indexed atom is shared by its
two adjacent blocks.
Hence
\[
 v(\tau_{2i+1})=n_i\quad(0\le i\le k+1),\qquad
 v(\tau_{2i})=n_{i-1}+n_i\quad(1\le i\le k+1).
\]

\begin{lemma}\label{l:stable-count-comparison}
Let $p$ and $q$ be multiples of $C$.
If $p+e=q+f$ with $0\le e,f<C$,
then $p=q$. Moreover, if positive integers $n_0,n_1,\ldots,n_N$
satisfy $n_i>D\sum_{h<i}n_h$ for every $i$, then a sum
$\sum_{i=0}^N d_i n_i$ with $d_i\in\{0,1,\ldots,D\}$ determines all its coefficients.
\end{lemma}
\begin{proof}
The multiple $p-q=f-e$ has absolute value less than $C$, so it is zero.
For the second assertion, take the largest index $i$ at which two
coefficient lists differ. Its contribution to the difference has absolute
value at least $n_i$, whereas all smaller indices together contribute
at most $D\sum_{h<i}n_h<n_i$.
\end{proof}

\begin{lemma}\label{l: comm stable}
For a data vector $w\le v$ and $\bsigma'\in L$ satisfying
\[
 w(\#)=0,\qquad \size w<C,\qquad \Par{\bsigma'}=v+w,
\]
there are integers $m_i\ge n_i$ $(0\le i\le k+1)$ such that
$\bsigma'$ is the word
\begin{equation*}
\begin{aligned}
(\tau_1\tau_2)^{m_0}
\#\#
(\tau_2\tau_3\tau_4)^{m_1}
&\#\#
(\tau_4\tau_5\tau_6)^{m_2}
\#\#
(\tau_6\tau_7\tau_8)^{m_3}\cdots \\[1ex]
&\#\#
(\tau_{2k}\tau_{2k+1}\tau_{2k+2})^{m_k}
\#\#
(\tau_{2k+2}\tau_{2k+3})^{m_{k+1}}\, .
\end{aligned}
\end{equation*}
or its reversal.
\end{lemma}

\begin{proof}

Put $v'=v+w$. Since $w\le v$ and $w(\#)=0$, the domain and separator
count are unchanged. Thus $\bsigma'$ also attains equality in
\Cref{p: dom bound}. It has $k+2$ blocks, with distinct atoms except
for the prescribed overlaps. Write $E'_j$ for the atom set of its $j$th
block and $c_j$ for its anchor. Only this block contains $c_j$, so its
counter is $v'(c_j)$. Counting the occurrences of $x\ne\#$ block by block gives
\[
 v(x)+w(x)=\sum_{j:x\in E'_j}v(c_j)+\sum_{j:x\in E'_j}w(c_j).
\]
The anchors are distinct, so the last sum, like $w(x)$, is at most
$\size w<C$. All the unperturbed counts are multiples of $C$.
By \Cref{l:stable-count-comparison},
\begin{equation}\label{eq:path-block-counts}
 v(x)=\sum_{j:x\in E'_j}v(c_j).
\end{equation}

Each $v(c_j)$ is one original counter or the sum of two adjacent ones.
At most two blocks contain $x$, so the coefficients on the right
of~\eqref{eq:path-block-counts} are at most two. Since
$n_i>2\sum_{h<i}n_h$, \Cref{l:stable-count-comparison} makes these
coefficients unique. Take $x=\tau_{2i+1}$. The left side is just $n_i$,
so there is exactly one summand, and it equals $n_i$.
The only atom of original count $n_i$ is $x$ itself. Hence $x$ is
the anchor of that block, whose new counter is $m_i=n_i+w(x)$.

For $x=\tau_{2i}$ the left side of~\eqref{eq:path-block-counts} is
$n_{i-1}+n_i$. Uniqueness now places $x$ in exactly the blocks with
anchors $\tau_{2i-1}$ and $\tau_{2i+1}$. These shared atoms determine
the path of blocks up to reversal. Once its direction is chosen, the
shared atoms occupy the adjoining ends of consecutive blocks, leaving
the anchor in the only remaining position. This gives the stated word
or its reversal, with every $m_i\ge n_i$.
\end{proof}

Thus $\bsigma$ is commutatively stable up to $C$.

\begin{corollary}
\label{c: comm stable}
Under the hypotheses of~\Cref{l: comm stable},
if $w(\tau_{2i+1})>0$ for some $i\in\{1,2,\ldots,k\}$,
then $w(\tau_{2i})>0$ and $w(\tau_{2i+2})>0$.
Similarly, if $w(\tau_1)>0$, then $w(\tau_2)>0$,
and if $w(\tau_{2k+3})>0$, then $w(\tau_{2k+2})>0$.
\end{corollary}

\begin{proof}
Put $\delta_i=m_i-n_i=w(\tau_{2i+1})\ge0$. Subtracting the old count
of each shared atom from its new count gives
\[
 w(\tau_{2i})=\delta_{i-1}+\delta_i\qquad(1\le i\le k+1).
\]
Thus $\delta_i>0$ increases both adjacent shared counts at an interior
block and the single adjacent shared count at either end.
\end{proof}

\subsection{Local propagation}

Define
\[
n=2M_0+2, \quad k=(r+1)N_0n^{r+1}, \quad C=(2k+3)N_0\, .
\]
Choose $\bsigma\in L^\prime$ of the form~\eqref{eq: word l} with
this value of $k$ and with counters $n_i=C\cdot 4^{i+1}$.
Let $v=\Par{\bsigma}$,
and fix a parsing tree $T_v$ of $v$ with respect to $R$.

We define the neighbor set $N(\tau_p)$
of $\tau_p$ by
$N(\tau_p)=\{\tau_{p-1},\tau_{p+1}\}$
for $p=2,3,\ldots,2k+2$,
$N(\tau_1)=\{\tau_2\}$,
and $N(\tau_{2k+3})=\{\tau_{2k+2}\}$.

\begin{claim}\label{c: odd implies neighbors}
For a vector $w$ in a non-root label of $T_v$,
\[
 \tau_{2j+1}\in\dom w,\quad \#\notin\dom w
 \quad\Longrightarrow\quad N(\tau_{2j+1})\subseteq\dom w.
\]
\end{claim}

\begin{proof}
Duplicate the proper subtree and prune the children of its copy. By \Cref{p: tree op}, $v+w\in\Par L$. Since $\size w\le N_0<C$ and $w(\#)=0$, apply \Cref{c: comm stable}.
\end{proof}

The following claim shows that saturated atoms
also propagate their values to
their neighbors.

\begin{claim}\label{c: sat implies neighbors}
For a subtree $T\subseteq T_v$ of value $w$,
\[
 w(\tau_i)=v(\tau_i)\quad\Longrightarrow\quad N(\tau_i)\subseteq\dom w.
\]
\end{claim}

\begin{proof}
For every $i$, $v(\tau_i)\ge4C>(v(\#)+1)N_0$.
The root and vertices contributing $\#$ account for at most
$(v(\#)+1)N_0$ copies. Choose a non-root vertex $u_i$ contributing
$\tau_i$ and no $\#$.
If $i$ is odd, saturation puts $u_i$ in $T$, and
\Cref{c: odd implies neighbors} applies.
If $i$ is even, that claim shows that $u_{i-1}$ and $u_{i+1}$ both
contribute $\tau_i$. Saturation puts both in $T$, supplying its two neighbors.
\end{proof}

\subsection{The counting argument}

We now bound the number of saturated atoms using the support of a single star sub-expression.

\begin{claim}
\label{c: many sep}
There is a vertex $u$ of $T_v$
with at least $n$ immediate children
whose sub-trees have nonzero $\#$-value.
\end{claim}

\begin{proof}
Otherwise, induct on $j$ to bound contributing subtrees at depth $j$ by $n^j$.
\emph{Base $j=0$.} There is one root.
\emph{Step.} Each contributing parent has fewer than $n$ contributing
children, so the induction hypothesis gives at most $n\cdot n^j=n^{j+1}$.
Every base contributes at most $N_0$, hence
\[
 v(\#)\le N_0\sum_{j=0}^r n^j\le(r+1)N_0n^r\le k<2k+2=v(\#),
\]
a contradiction.
\end{proof}

Let $R^\prime$ be the star sub-expression
corresponding to the vertex
given by~\Cref{c: many sep}.
Let $v_1,v_2,\ldots,v_n$
be the values of the sub-trees
rooted at $n$ distinct immediate children with $v_i(\#)>0$ for all $i=1,2,\ldots,n$.
Here $R^\prime$ denotes the indexed period form shared by these children. Put $S=\supp{R^\prime}\cup\{\#\}$, so $\card S\leq M_0$.

For each $j=1,2,\ldots,n$,
define 
\begin{align*}
A_j&= \set{a\in\dom{v_j}:v_j(a)=v(a)} \, ,\\
B_j&= \dom{v_j}\setminus S\, .
\end{align*}

\begin{claim}
The sets $A_j$ are pairwise disjoint,
and $\card{A_j}\,\geq v_j(\#)$ for $j=1,2,\ldots,n$.
\end{claim}

\begin{proof}
Disjoint sibling subtrees cannot both carry all copies of an occurring atom. Pruning child $j$ gives $v-v_j\in\Par L$ with domain $\dom v\setminus A_j$. By \Cref{p: dom bound},
\[
 \card{\dom v}-\card{A_j}\le v(\#)-v_j(\#)+2.
\]
Since $\card{\dom v}=v(\#)+2$, the asserted bound follows.
\end{proof}

\begin{claim}
For every $j=1,2,\ldots,n$,
$\card{B_j}\,\leq  v_j(\#)$
\end{claim}

\begin{proof}
Duplicate child $j$ and move $B_j$ in the copy to fresh atoms by an $S$-permutation. By \Cref{p: tree op} the new vector lies in $\Par L$. Its domain gains $\card{B_j}$ atoms and its separator count gains $v_j(\#)$ because $\#\in S$. Now \Cref{p: dom bound} and $\card{\dom v}=v(\#)+2$ give the inequality.
\end{proof}

\begin{proof}[Proof of \Cref{t: irrational}]
At most $M_0$ of the disjoint $A_j$ meet $S$. For every other $j$,
$A_j\subseteq B_j$ and $|A_j|\ge v_j(\#)\ge|B_j|$, so $A_j=B_j\ne\emptyset$.
Discard also the at most one $A_j$ containing $\tau_{2k+3}$.
At least $n-M_0-1=M_0+1$ remain. For each, choose maximal $t_j$ with
$\tau_{t_j}\in A_j$. By \Cref{c: sat implies neighbors},
\[
 \tau_{t_j+1}\in\dom{v_j}\setminus A_j
 =\dom{v_j}\setminus B_j\subseteq S.
\]
Disjointness makes the $t_j$ distinct, contradicting $|S|\le M_0$.
\end{proof}

\section{Quasi-context-free languages are not Parikh-equivalent to quasi-regular languages}
\label{a: par cfg neq par fma}

To establish this separation,
we define a finite-memory context-free grammar
with three registers that generates a language $L$
whose Parikh image cannot be matched by any
finite-memory automaton.

Here the blocks form a tree.
Commutative stability again recovers the atoms of each block from their counts, even when another derivation arranges the blocks in a different tree.

Let $\bG_3$ be a $\CFG_3$ with three fixed, distinct initial atoms different from $\#$, nonterminals $S,A,B,C,D,E$, and the following production rules:
\begin{align*}
S(x,y,z)&\to \#\#\#A(x,y,z),\\
A(x,y,z)&\to
xyz\,|\,\,
xyzA(x,y,z)\,|\,\,
xyzB(x,y,z)C(x,y,z),\\
B(x,y,z)&\to \#\#A(x,y^\prime,z^\prime) \, |\,
\# D(x,y^\prime,z) ,\\
C(x,y,z)&\to \#\#A(x^\prime,y^\prime,z) \, |\, 
\# E(x,y^\prime,z)\, ,\\
D(x,y,z)&\to xy\, | \, xy D(x,y,z) \, ,\\
E(x,y,z)&\to yz\, | \, yz E(x,y,z) \, .
\end{align*}

\begin{theorem}
\label{t: bg3 neq fma}
For every $\bA\in\FMA$ we have $\Par{L(\bA)}\neq \Par{\langof{\bG_3}}$.  
\end{theorem}

\subsection{Commutative stability in context-free grammars}

Each derivation of a word in $\langof{\bG_3}$ induces a binary tree. A run of $A(x,y,z)$ gives one block $(xyz)^n$ at a triple node $(x,y,z)$, with counter $n$. The node may terminate there or branch to $BC$. Its left child is a triple from $A$ or a pair from $D$. Its right child is a triple from $A$ or a pair from $E$. Each pair leaf likewise emits one block, whose counter is the number of pair repetitions.

In a branching triple $(a,b,c)$, the left child inherits $a$ on its left and the right child inherits $c$ on its right.
\begin{example}
Let $a,b,\ldots,h$ be distinct atoms different from $\#$, with $(a,b,c)$ the initial tuple. For positive $n_1,n_2,\ldots,n_5$, consider the following rule instances.
\begin{align*}
\conf{S}{\tuple{a,b,c}}&\Longrightarrow \#^3
\conf{A}{\tuple{a,b,c}}\, , \\
\conf{A}{\tuple{a,b,c}}&\Longrightarrow^\ast 
(abc)^{n_1} \conf{B}{\tuple{a,b,c}}\conf{C}{\tuple{a,b,c}}\, , \\
\conf{B}{\tuple{a,b,c}}&\Longrightarrow
\# \conf{D}{\tuple{a,d,c}} \Longrightarrow^\ast 
\# (ad)^{n_2} \, , \\
\conf{C}{\tuple{a,b,c}}&\Longrightarrow
\#^2 \conf{A}{\tuple{e,f,c}} \, . \\
\conf{A}{\tuple{e,f,c}}&\Longrightarrow^\ast 
(efc)^{n_3} \conf{B}{\tuple{e,f,c}}\conf{C}{\tuple{e,f,c}}\, . \\
\conf{B}{\tuple{e,f,c}}&\Longrightarrow
\# \conf{D}{\tuple{e,g,c}} \Longrightarrow^\ast 
\# (eg)^{n_4} \, , \\
\conf{C}{\tuple{e,f,c}}&\Longrightarrow
\# \conf{E}{\tuple{e,h,c}} \Longrightarrow^\ast 
\# (hc)^{n_5} \, ,
\end{align*}
They give the word
\[
\btau= \#^3 (abc)^{n_1} \# (ad)^{n_2}
\#^2 (efc)^{n_3} \# (eg)^{n_4}
\# (hc)^{n_5} \, .
\]
It has the following Parikh vector:
\[
\Par{\btau} =\ 8 \#+(n_1+n_2) a
+n_1 b+(n_1+n_3+n_5) c
+n_2 d+(n_3+n_4) e+n_3 f+n_4 g+n_5 h\, .
\]
Note that each counter $n_i$
is the value of some atom.
The tree of $\btau$ is depicted in \Cref{fig: gen tree}.
\end{example}

\begin{figure}[htbp]
\centering
\begin{tikzpicture}[scale=0.8,
level distance=16mm,
sibling distance=28mm,
every node/.style={font=\small, align=center, inner sep=2pt},
nt/.style={rectangle, draw, rounded corners, minimum width=18mm},
term/.style={rectangle, draw, inner sep=2pt},
dots/.style={font=\large}
]

\node[nt]{$\colorbox{yellow}{a},b,\colorbox{cyan}{c}$\\ $n_1$}
child { node [nt]{$\colorbox{yellow}{a},d$\\$n_2$}}
child {
node[nt]{$\colorbox{lime}{e},f,\colorbox{cyan}{c}$\\$n_3$}
child { node[nt]{$\colorbox{lime}{e},g$\\$n_4$}}
child { node[nt]{$h,\colorbox{cyan}{c}$\\$n_5$} }
};

\end{tikzpicture}
\caption{Natural tree associated with the generation
of $\btau$.}
\label{fig: gen tree}
\end{figure}

We say that node $u$
is a left-most ancestor of node $u^\prime$
if there is a path $u_0,u_1,\ldots,u_n$
such that $u_0=u$, $u_n=u^\prime$,
and $u_{i+1}$ is the left child of $u_i$ for $i=0,1,\ldots,n-1$.
Symmetrically, we have right-most ancestors.
For an internal node labeled $(a,b,c)$,
we call the middle atom $b$
the \emph{anchor} of this node. Use the middle atom also for a terminal triple. In a pair leaf the newly introduced atom is its anchor.

\begin{lemma}
\label{l: bg3 tree}
Let $\btau$ be a word in $\langof{\bG_3}$
with $T$ its tree and $v=\Par{\btau}$
its Parikh vector.
Then,
\begin{enumerate}
\item If $u$ is a  left-most ancestor of $u^\prime$,
then the left atom of $u$ and the left atom of $u^\prime$ are equal.
Symmetrically for right-most ancestors.
\item $\card{\dom{v}}\, \leq 1+v(\#)$.
\item If $\card{\dom{v}}\,= 1+v(\#)$, then
equal occurrences in different node labels are precisely copies of the same introduced atom along a left-most or right-most chain.

\item The number of nodes in $T$ is at most $3v(\#)$.
\end{enumerate}
\end{lemma}

\begin{proof}
The $B,C$ rules give inheritance. If $I$ counts branching triples and
$J$ terminal triples, the pair-leaf count is $I+1-J$ and the node count
is $2I+1$. Structural separators and introduced positions both number
\[
 H=3+2(I+J-1)+(I+1-J)=3I+J+2.
\]
Thus $|\dom v|\le1+H\le1+v(\#)$ and $2I+1\le3v(\#)$.
Equality forces all introductions to be distinct and different from $\#$,
leaving precisely the inherited repetitions.
\end{proof}

We restrict attention to words
that use the maximum possible
number of distinct atoms,
that is $\card{\dom{v}}\,= 1+v(\#)$,
and call them
\emph{maximal-distinct} words.
In a maximal-distinct word each designated anchor belongs to just one block, so its value in the Parikh vector equals that block's counter. In the full binary tree used below, each non-anchor atom other than $\#$ occurs in at least two node labels: at an internal node its outside atoms are inherited by its children, and at a pair leaf its non-anchor atom is inherited from its parent.

For a word $w\in\{0,1\}^\ast$, let $\mathrm{slex}(w)$
be the position of $w$
in the length-lexicographical order (shortlex)
of all binary strings. The first five values are
\[
\mathrm{slex}(\epsilon)=0, \ \mathrm{slex}(0)=1, \ \mathrm{slex}(1)=2,
\ \mathrm{slex}(00)=3,\ \mathrm{slex}(01)=4.
\]

Fix positive integers $C,M,d$. 
Let $n_w=C\cdot M^{\mathrm{slex}(w)+1}$
be the counters for each vertex by the shortlex ordering.
That is, $n_\varepsilon=C\cdot M$, $n_0=C\cdot M^2$, $n_1=C\cdot M^3$.

Let $\btau\in \langof{\bG_3}$ be a
maximal-distinct word
whose tree is a full binary tree
of depth $d\ge1$, with branching
triples at internal nodes, pair leaves, and counters $n_w$.
Choose every newly introduced atom fresh, except for the prescribed root tuple.
Let
\[
v=\Par{\btau} \qquad \mbox{and} \qquad
k=\,\card{\dom{v}}\,\leq 3\cdot 2^{d+1}<2^{d+3}\, .
\]
Notice, from maximal-distinctness, $v(\#)=k-1$.
Let $T$ be the tree of $\btau$.

\begin{example}
The tree
for $d=2$ is depicted in~\Cref{fig: gen tree btau},
where atoms marked by the same color are equal.

\begin{figure}[htbp]
\centering
\begin{tikzpicture}[scale=0.6,
level distance=16mm,
level 1/.style={sibling distance=80mm},
level 2/.style={sibling distance=36mm},
every node/.style={font=\small, align=center, inner sep=2pt},
nt/.style={rectangle, draw, rounded corners, minimum width=18mm},
term/.style={rectangle, draw, inner sep=2pt},
dots/.style={font=\large}
]

\node[nt]{$\colorbox{yellow}{$a_\epsilon$},b_\epsilon
,\colorbox{cyan}{$c_\epsilon$}$\\ $n_\epsilon$}
child { node [nt]{$\colorbox{yellow}{$a_0$},b_0,\colorbox{pink}{$c_{0}$}$\\$n_0$}
child { node [nt]{$\colorbox{yellow}{$a_{00}$},b_{00}$\\ $n_{00}$}}
child { node [nt]{$b_{01}
,\colorbox{pink}{$c_{01}$}$\\ $n_{01}$}}}
child {
node[nt]{$\colorbox{lime}{$a_{1}$},b_1,\colorbox{cyan}{$c_1$}$\\$n_1$}
child { node[nt]{$\colorbox{lime}{$a_{10}$},b_{10}$\\$n_{10}$}}
child { node[nt]{$b_{11},\colorbox{cyan}{$c_{11}$}$\\$n_{11}$} }
};
\end{tikzpicture}
\caption{Natural tree associated with the generation
of $\btau$.}
\label{fig: gen tree btau}
\end{figure}
\end{example}

Write $n_u$ for the counter at $u\in T$ and $n'_v$ for the counter at $v$
in any other derivation tree $T'$.

\begin{lemma}\label{l: comm stable tree}
Suppose $M>3k+1$, $w\le v$, $w(\#)=0$, $\size w<C$, and
$\btau'\in\langof{\bG_3}$ satisfies $\Par{\btau'}=v+w$.
Every tree $T'$ of $\btau'$ admits a bijection $\beta:T\to T'$ preserving
node atom sets and anchors, with
\[
 n'_{\beta(u)}\ge n_u\qquad(u\in T).
\]
\end{lemma}

\begin{proof}

Put $v'=v+w$. Its domain and separator count agree with those of $v$,
so $T'$ is maximal-distinct. Let $I,I'$ count the branching triples
in $T,T'$, respectively, and let $J'$ count the terminal triples in $T'$.
There are none in $T$. The separator count in the proof of
\Cref{l: bg3 tree} gives $3I+2=3I'+J'+2$. Consequently
$|T'|=2I'+1\le 2I+1=|T|$.

Let $E_u$ be the atom set at $u\in T$. For $p\in T'$ write $E'_p$
for its atom set and $c_p$ for its anchor. Its counter is $v'(c_p)$.
Counting $x\ne\#$ in $T'$ gives $v'(x)=\sum_{p:x\in E'_p}v'(c_p)$.
The anchors $c_p$ are distinct, so the increments on each side total
less than $C$. Removing them by \Cref{l:stable-count-comparison} gives
\begin{equation}\label{eq:tree-block-counts}
 \sum_{u:x\in E_u}n_u=v(x)=\sum_{p:x\in E'_p}v(c_p),
 \qquad v(c_p)=\sum_{u:c_p\in E_u}n_u .
\end{equation}
After expanding the right side, each coefficient is at most
$|T'|\le3v(\#)<3k$ by \Cref{l: bg3 tree}. In increasing order, each
counter exceeds $3k$ times the sum of its predecessors since $M>3k+1$.
Thus the coefficients on the two sides agree by \Cref{l:stable-count-comparison}.
For the anchor $x$ at $u$, the left side is $n_u$ alone.
Hence exactly one $p$ contributes, and its $c_p$ belongs only to $u$
in $T$. The only such atom is $x$, so $c_p=x$. Set $\beta(u)=p$.

Distinct anchors give distinct nodes $\beta(u)$, and $|T'|\le|T|$
makes $\beta$ a bijection. In~\eqref{eq:tree-block-counts} the summand
for $\beta(u)$ is now $v(c_{\beta(u)})=n_u$. Comparing its coefficient
for each $x$ gives $x\in E_u$ exactly when $x\in E'_{\beta(u)}$.
Finally, if $b$ is the anchor of $u$, its new counter is
$n'_{\beta(u)}=v'(b)=n_u+w(b)\ge n_u$.
\end{proof}

\begin{corollary}
\label{c: comm stable tree anchors}
Under the hypotheses of~\Cref{l: comm stable tree},
if $b$ is an anchor of the triple $(a,b,c)$
and $w(b)>0$, then $w(a)>0$ and $w(c)>0$.
\end{corollary}

\begin{proof}
Both outside atoms receive the increment $w(b)$ from this node and nonnegative increments elsewhere.
\end{proof}

\subsection{Proof strategy}

Suppose, for a contradiction, that an automaton $\bA\in\FMA$
satisfies $\Par{L(\bA)}=\Par{\langof{\bG_3}}$.
For the remainder of this section, 
let $r_{A}$ and $N_A$ be the number of registers and states in $\bA$.

Choose $d$ first, then $M>3k+1$, and finally $C$ sufficiently large. Apart from the fixed root atoms choose all introduced data outside $D$, the initial register tuple of $\bA$, and $\{\#\}$. 

Since $\Par{L(\bA)}= \Par{\langof{\bG_3}}$,
there is a word $\bsigma\in L(\bA)$
with the same Parikh vector as $\btau$.
Let
$\brho=\conf{s_0}{\br_0},
\conf{s_1}{\br_1},\ldots,\conf{s_m}{\br_m}$
be an accepting run of $\bA$
on $\bsigma$.

For an atom $a\in [\bsigma]\setminus\set{\#}$,
let $I_a$
be the closed interval $[i,j-1]$, where $i,j$ are the first and last positions of $a$ in $\bsigma$, numbered from one. Its integer points are cuts with an occurrence on both sides.
The rest of the proof has three steps.
\begin{enumerate}
\item At every integer cut in $I_a$, the automaton stores $a$ if $a\notin D$.
Otherwise, renaming $a$ in the suffix would violate the bound on distinct atoms.
\item For each triple $(a,b,c)$, pumping a short contribution containing $b$
and applying commutative stability give $I_a\cap I_c\ne\emptyset$.
\item The inherited pairs of a sufficiently large full binary tree then force
at least $r_A+4$ intervals through a common integer cut. At most three of their
atoms are prescribed root atoms, and all other atoms lie outside $D$.
\end{enumerate}
The first step therefore requires more than $r_A$ distinct atoms in the registers
at that cut, contradicting the register bound of $\bA$.

\subsection{Intervals preserve their atoms}

\begin{lemma}\label{l: interval storage}
For all $a\in[\bsigma]\setminus(\{\#\}\cup D)$ and integer cuts $i\in I_a$, we have $a\in[\br_i]$.
\end{lemma}
\begin{proof}
If $a\notin[\br_i]$, transpose $a$ with an atom outside $[\bsigma]\cup[\br_i]\cup D$ in the suffix run. The cut configuration is fixed, so the new run accepts by \Cref{p: permutation word fma}. The prefix retains $a$, and the suffix acquires the fresh atom. Thus the domain grows by one without adding a separator, contradicting \Cref{l: bg3 tree}.
\end{proof}

\subsection{Intervals of blocks intersect}

Put
\begin{equation}\label{eq:pumping-bound}
 Q=N_A(k+|D|+2r_A)^{r_A}.
\end{equation}
An \emph{infix} consists of consecutive positions. A \emph{subword} is obtained
by deleting positions and preserving the order of the remaining letters.

We use the repeated-configuration case of the pumping argument
in~\cite[Lemma~24]{DanieliPump26}.

\begin{lemma}\label{l: pump one}
Every infix of $\bsigma$ with more than $Q$ occurrences of $\tau$ has a
subword $\bomega$ with $\Par{\bomega}(\tau)>0$ and $|\bomega|\le Q$ such that
inserting $\bomega$ at a cut in that infix gives a word in $L(\bA)$.
\end{lemma}
\begin{proof}
Put $B=[\bsigma]\cup D$. Choose $E\subseteq\atoms\setminus B$ of size $2r_A$
containing $[\br_0]\setminus B$. An accepting run can use only $B\cup E$.
Indeed, induct on its transitions. \emph{Base.} Keep $\br_0$.
\emph{Step.} Fix $B$ and retain the assigned images of current register values.
Assign each new value outside $B$ to a distinct unused member of $E$.
At most $r_A$ members are occupied and at most $r_A$ are needed, preserving
all equalities in the transition constraint.

There are at most $Q$ configurations over $B\cup E$. Among more than $Q$
occurrences of $\tau$, two preceding configurations coincide.
The intervening run starts with a $\tau$-transition. Keep it and erase cycles
from the return path, obtaining a loop of length at most $Q$ whose label
$\bomega$ is a subword of the infix. Insert this loop at its first configuration.
The original suffix still accepts.
\end{proof}

\begin{lemma}\label{l: intval intsect}
If $C>(k+4)Q$, then
\[
 I_a\cap I_c\ne\emptyset\qquad\text{for every node $(a,b,c)$ of $T$}.
\]
\end{lemma}
\begin{proof}
Mark the $k-1$ separators and the first and last occurrences of $a,c$.
The at most $k+4$ intervening infixes contain all $v(b)\ge C>(k+4)Q$
copies of $b$, so one contains more than $Q$.
By \Cref{l: pump one}, inserting its subword adds a vector $w\le v$ with
$w(b)>0$, $w(\#)=0$, and $\size w\le Q<C$.
Equality of the Parikh images and \Cref{c: comm stable tree anchors}
give $w(a),w(c)>0$. The infix therefore contains both atoms and none of their marked extremes,
so its cuts lie between both pairs of extremes. Hence $I_a\cap I_c\ne\emptyset$.
\end{proof}

\subsection{Interval-intersecting trees}

\begin{definition}\label{d: interval trees}
An \emph{interval-intersecting tree} with $d\ge1$ levels is a full ordered
binary tree $T_d$, with root at level one and all leaves at level $d$, such that:
\begin{enumerate}
\item Each node carries a pair of labels. All labels at the leaves are distinct.
\item A node with left-child pair $(a,b)$ and right-child pair $(c,e)$ carries $(a,e)$.
\item Each label $a$ has a nonempty closed bounded real interval $I_a$, and
$I_a\cap I_b\ne\emptyset$ for every node pair $(a,b)$.
\end{enumerate}
Let $f(d)$ be the least possible maximum number of distinct label intervals
that have a common intersection point.
\end{definition}

\begin{example}
\Cref{fig:tree 3,fig: int t3} give an interval-intersecting tree with three
levels and maximum overlap two.
Since each node pair has distinct labels, $f(2)=f(3)=2$.

\begin{figure}[H]
\centering
\begin{tikzpicture}[scale=0.9,>=stealth, every node/.style={font=\small}]

\definecolor{Acol}{RGB}{220,20,60}
\definecolor{Bcol}{RGB}{30,144,255}
\definecolor{Ccol}{RGB}{34,139,34}
\definecolor{Dcol}{RGB}{255,140,0}
\definecolor{Ecol}{RGB}{148,0,211}
\definecolor{Fcol}{RGB}{0,128,128}
\definecolor{Gcol}{RGB}{139,69,19}
\definecolor{Hcol}{RGB}{255,105,180}

\node (v1) at (0,0)      {$(A,B)$};
\node (v2) at (-3,-2)    {$(A,C)$};
\node (v3) at ( 3,-2)    {$(D,B)$};
\node (v4) at (-4,-4)    {$(A,E)$};
\node (v5) at (-2,-4)    {$(F,C)$};
\node (v6) at ( 2,-4)    {$(D,G)$};
\node (v7) at ( 4,-4)    {$(H,B)$};

\draw[thick] (v1) -- (v2);
\draw[thick] (v1) -- (v3);
\draw[thick] (v2) -- (v4);
\draw[thick] (v2) -- (v5);
\draw[thick] (v3) -- (v6);
\draw[thick] (v3) -- (v7);

\end{tikzpicture}
\caption{Labels for $T_3$.}
\label{fig:tree 3}
\end{figure}

\begin{figure}[H]
\centering
\begin{tikzpicture}[scale=1,>=stealth, every node/.style={font=\small}]

\definecolor{Acol}{RGB}{220,20,60}
\definecolor{Bcol}{RGB}{30,144,255}
\definecolor{Ccol}{RGB}{34,139,34}
\definecolor{Dcol}{RGB}{255,140,0}
\definecolor{Ecol}{RGB}{148,0,211}
\definecolor{Fcol}{RGB}{0,128,128}
\definecolor{Gcol}{RGB}{139,69,19}
\definecolor{Hcol}{RGB}{255,105,180}

\draw[Acol, ultra thick] (2.0,4.0) -- (6.0,4.0);
\node[left] at (4,4.0) [above] {$A$};

\draw[Bcol, ultra thick] (5.0,3.5) -- (9.0,3.5);
\node[left] at (7,3.5) [above] {$B$};

\draw[Ccol, ultra thick] (1.0,3.5) -- (2.75,3.5);
\node[left] at (2,3.5) [above] {$C$};

\draw[Ecol, ultra thick] (3.25,3.5) -- (3.75,3.5);
\node[left] at (3.5,3.5) [above] {$E$};

\draw[Fcol, ultra thick] (0.5,3) -- (1.5,3);
\node[left] at (1,3) [above] {$F$};

\draw[Dcol, ultra thick] (8,3) -- (10,3);
\node[left] at (9,3) [above] {$D$};

\draw[Hcol, ultra thick] (6.5,3) -- (7.5,3);
\node[left] at (7,3) [above] {$H$};

\draw[Gcol, ultra thick] (9.5,2.5) -- (11.5,2.5);
\node[left] at (10.5,2.5) [above] {$G$};

\end{tikzpicture}
\caption{Orientation of intervals of $T_3$ with at most
two intersecting intervals.}
\label{fig: int t3}
\end{figure}
\end{example}

Let $B_0$ be a single vertex. For $j\ge1$, form the rooted tree $B_j$
by joining a new root to the roots of disjoint copies of
$B_0,B_1,\ldots,B_{j-1}$.

The graph whose vertices are the labels and whose edges are the pairs
of labels at the nodes, rooted at either top label, is isomorphic to $B_d$.

\begin{definition}[{\cite[Definition~1]{BatzillBiedlPath}}]\label{d: tree pathwidth}
The \emph{pathwidth} $p(U)$ of a tree $U$ is zero if $U$ is a single
vertex. Otherwise,
\[
 p(U)=\min_{P\text{ a path in }U}
 \max\bigl(\{1\}\cup\{1+p(C):C\text{ is a component of }U-P\}\bigr),
\]
where $U-P$ is obtained by deleting the vertices of $P$.
\end{definition}

Equivalently, assign a nonempty closed bounded interval to each vertex
of $U$ so that the intervals of adjacent vertices meet.
Then $p(U)+1$
is the least possible maximum number of intervals containing a common
point~\cite[Theorem~29]{BodlaenderArboretum}.

Put $p_j=p(B_j)$. Applying this characterization to the graph on
labels gives $f(d)=1+p_d$.

\begin{theorem}\label{t: low int tree}
For every $d\ge1$, $f(d)=1+\lceil d/3\rceil$.
\end{theorem}
\begin{proof}
Fix $j\ge3$. A path in $B_j$ meets at most two branches at the root,
so it misses at least one of the three largest branches. Each of these
contains a copy of $B_{j-3}$. Thus deleting any path leaves a component
of pathwidth at least $p_{j-3}$, giving $p_j\ge1+p_{j-3}$.

For the reverse inequality, three successive joins express $B_j$ as
eight copies of $B_{j-3}$ whose roots form the label graph in
\Cref{fig:tree 3}. Deleting the path $C,A,B,D$ removes every edge
between distinct copies. Each remaining component lies in a single
copy and therefore has pathwidth at most $p_{j-3}$.
By \Cref{d: tree pathwidth}, $p_j\le1+p_{j-3}$.

We conclude by strong induction. The trees $B_0,B_1,B_2$ are paths,
so $(p_0,p_1,p_2)=(0,1,1)$. For $j\ge3$, the recurrence and the
induction hypothesis give
\[
 p_j=1+p_{j-3}=1+\lceil(j-3)/3\rceil=\lceil j/3\rceil.
\]
Since $f(d)=1+p_d$, the result follows.
\end{proof}

\begin{proof}[Proof of~\Cref{t: bg3 neq fma}]
Take $d=3r_A+7$, fix the resulting $k$, and choose $M>3k+1$ and $C>(k+4)Q$, with $Q$ from~\eqref{eq:pumping-bound}.
Delete the pair leaves and retain the two outside atoms of each triple.
Maximal-distinctness gives distinct bottom labels. Inheritance and
\Cref{l: intval intsect} give an interval-intersecting tree with $d$ levels
as in \Cref{d: interval trees}.
By \Cref{t: low int tree}, at least $r_A+4$ intervals meet. Their largest left
endpoint is a common integer cut. At most three labels are prescribed root atoms,
so more than $r_A$ atoms outside $D$ must be stored there, a contradiction.
\end{proof}

\section{Concluding remarks}
\label{s: conc}

Several natural questions remain open.
Most notably, the case of \emph{two registers}
remains unresolved,
both for finite-memory automata
and for context-free grammars.
At present, no counterexample
is known for a commutatively stable
language generated by either model.
It is therefore unclear whether these
models always admit rational Parikh images
or whether they generate irrational images
that nevertheless coincide.
More generally, given a commutatively stable
language, the techniques developed in this paper
appear robust enough to establish irrationality and Parikh
in-equivalence between
grammars and automata,
should suitable counterexamples exist.

We mention several decision problems.
For one-register automata,
star-height zero of the Parikh image
is easily decidable, as it coincides
with boundness.
This raises
the question of whether it is decidable if
the Parikh image has star-height one.
A positive answer would yield a complete
decision procedure
for determining the exact star-height
of Parikh images of one-register automata.

\bibliographystyle{plainurl}
\bibliography{references}

\end{document}